\documentclass[aip,jmp,numerical,reprint,eqsecnum,floats,groupedaddress,showpacs,notitlepage,onecolumn]{revtex4-1}

\usepackage{amsmath}
\usepackage{amssymb}
\usepackage{mathrsfs}
\usepackage{amsthm}
\usepackage{textcomp}
\usepackage{xcolor}
\usepackage{graphicx}
\usepackage{tikz}
 \usetikzlibrary{decorations.markings}
\usepackage{enumerate}
\usepackage{mathtools}
\usepackage{hyperref}
\usepackage[title]{appendix}
\usepackage[caption=false]{subfig}
\usepackage{bm}
\usepackage{natbib}
\usepackage[left=3cm,right=3cm]{geometry}

\newtheorem{thm}{Theorem}[section]
\newtheorem{lem}[thm]{Lemma}
\newtheorem{prop}[thm]{Proposition}
\newtheorem{cor}[thm]{Corollary}
\newtheorem{defi}[thm]{Definition}
\newtheorem*{remark}{Remark}
\newtheorem*{remarks}{Remarks}

\newcommand{\A}{\mathcal{A}}
\newcommand{\C}{\mathcal{C}}
\newcommand{\D}{\mathcal{D}}

\newcommand{\eps}{\varepsilon}
\newcommand{\iStar}{{\phantom{*}}}
\newcommand{\iPower}{{\phantom{2}}}
\renewcommand{\P}[1]{P_{{#1},{#1+1}}^*}
\newcommand\sumcos[1]{\sum_{j=1}^{#1}\prod_{\substack{k=1 \\ k\neq j}}^{#1}c_k^\iPower}

\begin{document}

 \title{Multistability of Phase-Locking and Topological Winding Numbers 
 in Locally Coupled Kuramoto Models on Single-Loop Networks}.
 \date{14 December 2015, revised: 15 February 2016, published: 9 March 2016}
 \author{Robin Delabays}
 \author{Tommaso Coletta}
 \author{Philippe Jacquod}
 \affiliation{School of Engineering, University of Applied Sciences of Western Switzerland, CH-1950 Sion, Switzerland}

 \keywords{Kuramoto model, multistability, winding numbers, loop flows}
 \pacs{05,45.-a, 05.45.Xt, 84.70.+p}
  
 \begin{abstract}
 \textit{
 Determining the number of stable phase-locked solutions for locally coupled Kuramoto models
 is a long-standing mathematical problem with important implications in biology, condensed matter  
 physics and electrical engineering among others. We investigate
 Kuramoto models on networks with various topologies and show that different phase-locked 
 solutions are related to one another by loop currents. The latter take only discrete values, 
 as they are characterized by topological winding numbers. This result is generically valid
 for any network, and also applies beyond the Kuramoto model, 
 as long as the coupling between oscillators is antisymmetric in the oscillators'
 coordinates. Motivated by these results we further investigate loop currents in Kuramoto-like models. 
 We consider loop currents in  nonoriented $n$-node
 cycle networks with nearest-neighbor coupling. 
 Amplifying on earlier works, we give an algebraic upper bound $\mathcal{N} \le 
 2 \, {\rm Int}[n/4]+1$ for the number $\cal N$ of different, linearly stable 
 phase-locked solutions. We show that the number of different 
 stable solutions monotonically decreases as the 
 coupling strength is decreased. 
 Furthermore stable solutions with a single angle difference exceeding $\pi/2$ emerge 
 as the coupling constant $K$ is reduced, 
 as smooth continuations of solutions with all angle differences smaller than $\pi/2$ at higher $K$.
 In a cycle network with nearest-neighbor coupling 
 we further show that phase-locked solutions with two or more angle differences larger
 than $\pi/2$ are all linearly unstable. We point out similarities between loop currents 
 and vortices in superfluids and superconductors as well as 
 persistent currents in superconducting rings and two-dimensional Josephson junction arrays. 
 }
 \end{abstract}

 \maketitle

\section{Introduction}
From large colonies of fireflies flashing in unison to single-frequency electric power grids
covering areas as large as entire continents, from human brain waves to arrays of submicronic 
Josephson junctions, there are many, disparate systems that exhibit collective 
synchrony~\cite{Str03}. Following early works, most notably by Winfree~\cite{Win67},
a window towards a quantitative, mathematical understanding of collective synchrony
was opened by Kuramoto~\cite{Kur84} who 
proposed a model of coupled oscillators defined by the following set of 
nonlinear differential equations
\begin{eqnarray}
\label{eq:Kuramoto}
\dot{\theta}_i = P_i - \sum_{j=1}^n K_{ij} \sin(\theta_i - \theta_j) \, , \,\,\,\,\,\,\, i=1,\ldots n\, .
\end{eqnarray}
The model describes the dynamics of a set of $n$ one-dimensional 
oscillators with angular coordinates $\theta_i$ and natural  
frequencies $P_i$ under the influence of a 
coupling that is periodic in their angle differences. 
The Kuramoto model has 
become a standard model for investigating the transition to
synchrony  in coupled dynamical systems~\cite{Str00,Ace05,Dor14}. 

The beauty of the Kuramoto model is that it is sufficiently simple to allow for analytical treatments
of the emergence of synchrony in coupled oscillators systems, while retaining most of the essence of 
this complex problem. As a matter of fact, 
Kuramoto observed early on that for constant all-to-all coupling, $K_{ij} \equiv K/n$, 
an analytically solvable mean-field solution becomes exact in the large $n$ limit. 
A coherent, synchronous state $\{ \theta_i^{(0)} \}$
emerges for $K > K_c$
in the form of a solution to Eq.~(\ref{eq:Kuramoto}) with $\dot\theta_i^{(0)}-\dot\theta_j^{(0)}=0$, 
for at least a finite fraction of pairs of oscillators $(i,j)$. The critical coupling strength $K_c$
depends on the
distribution $g(P)$ of natural frequencies $P_i$, and 
phase-locking with $\dot\theta_i^{(0)}-\dot\theta_j^{(0)}=0$, for all $i,j$, can be achieved
if $g(P)$ has compact support~\cite{Erm85,Hem93}. 

Most physical systems exhibiting synchrony consist however in collections of 
subsystems with short-range coupling. The problem becomes much more complicated 
for Kuramoto models defined on such complex networks with reduced node degree
because the mean-field approach no longer applies. Phase-locked solutions $\{ \theta_i^{(0)} \}$ to 
Eq.~(\ref{eq:Kuramoto}) are determined by
\begin{eqnarray}\label{eq:Ksteady}
P_i = \sum_{j=1}^n K_{ij} \sin(\theta_i^{(0)} - \theta_j^{(0)}) \, , 
\end{eqnarray}
i.e. by a set of $n$ nonlinear algebraic equations which, in principle, accept more than one solution. 
For instance new solutions to Eq.~(\ref{eq:Ksteady}) can be obtained from known
solutions by substituting $\theta_i^{(0)}  - \theta_j^{(0)}  \rightarrow \pi-(\theta_i^{(0)}  - \theta_j^{(0)} )$
for some or all $(i,j)$.
This can lead, in principle, to an exponential number $\propto 2^{n}$ of solutions, however, many of them
are not dynamically stable in the sense given by Eq.~\eqref{eq:Kuramoto}. 
It has in particular been shown that there is a single stable solution above 
the transition to synchrony
for all-to-all couplings~\cite{Aey04,Mir05} and for
identical oscillators ($P_i=\bar P$, for all $i$) on networks with sufficiently 
large node degree~\cite{Tay12}.

To be physically relevant, a solution 
of Eq.~(\ref{eq:Ksteady}) needs to be robust against any small perturbation.
Thus, the truly important question is "how many {\it linearly stable} 
phase-locked solutions to the Kuramoto problem are there ?" 
This question dates back at least to the work of 
Korsak~\cite{Kor72} in the context of the power flow 
problem (dealing in particular with conditions for operational 
synchrony in the electric power grid)~\cite{Ber00}, which is 
closely related to the Kuramoto model. 
As a matter of fact, it turns out (see below in Section~\ref{ssec:powerflow})
that for high voltages, a first 
approximation is to neglect ohmic losses, in which case AC electric 
power transport between the nodes of a power grid is 
governed by Eq.~(\ref{eq:Ksteady}), with $P_i$ being the power injected ($P_i>0$)
or extracted ($P_i<0$) at node $i$. Korsak provided a simple example of a network where 
different, linearly stable solutions exist that differ by some circulating loop current.
Similar works have dealt with that problem since then, motivated by issues of voltage and phase
stability that are central to the stable, synchronous operation of electric power grids. Most, if not all of 
these investigations are however
restricted to numerical investigations on small networks. The literature on the subject is rather large and
we refer the interested reader to Refs.~\onlinecite{Tam83,Klos91} for more information.
Bounds for the number of different stable solutions were first constructed in the spirit of the argument
given after Eq.~(\ref{eq:Ksteady}). In this way,
Refs.~\onlinecite{Bai82} and \onlinecite{Ngu14} gave an exponential upper bound for the number of power flow solutions,
as did Ref.~\onlinecite{Meh14} for the phase-locked solutions of the Kuramoto model. 
Ref.~\onlinecite{Meh14} observed numerically, however, that the number of stable solutions is much smaller than $2^n$.
Below we show that a much better, algebraic upper bound is obtained when considering 
that quantized loop currents differentiate between different stable solutions. 

To the best of our knowledge, 
the characterization of loop flows with topological winding numbers has been first made by 
Janssens and Kamagate~\cite{Jan03}, though L\"uders  (in a referee discussion
at the end of Ref.~\onlinecite{Kor72}) and Ermentrout~\cite{Erm85}
already point to it. Topological winding numbers
emerge from the consistency requirement 
that summing angle differences along any cycle in a network must give an integer
multiple of $2 \pi$.
Below we illustrate how this leads to loop currents that can take only discrete values. 
Similar considerations in a different physical context lead to the quantization of circulation around 
vortices in a superfluid~\cite{Ons49,Fey55} or a type-II superconductor~\cite{Abr57}
and to the quantization of persistent currents in superconducting rings~\cite{Bye61,Mat02} and
rings of Josephson junctions~\cite{Ras13}. That such a similarity exists is not a surprise,
given that each term on the right-hand-side of Eq.~(\ref{eq:Ksteady}) gives the Josephson
current between two superconductors with order parameter 
$\psi_{i,j}=n_s^{1/2} \exp i \theta_{i,j}^{(0)}$
coupled by a tunnel junction of transparency $T_{ij} = \hbar K_{ij}/8 e n_s$. 
More surprising, however, is that
investigations on small cycle networks with injections and consumptions show that
loop currents persist even in networks with ohmic dissipation~\cite{Kor72,Buk13}.

Using winding numbers, Rogge and Aeyels~\cite{Rog04}
obtained an algebraic upper bound $\mathcal{N} \le 2\, {\rm Int}[n/4]+1$ 
for the number of stable solutions with any angle difference in a Kuramoto model on a 
$n$-node ring with unidirectional nearest-neighbor couplings.
The same upper bound has been calculated 
by Ochab and G\'ora~\cite{Och10} in a nonoriented $n$-node Kuramoto ring with 
nearest-neighbor coupling, under the condition that 
all angle differences are smaller than $\pi/2$. This upper bound is reached when the 
coupling strength goes to infinity, equivalently corresponding to $P_i=0$, for all $i$ in 
Eqs.~(\ref{eq:Kuramoto}) and (\ref{eq:Ksteady}), i.e. to identical oscillators. This alternatively gives 
the number of stable states for Josephson junction rings in the classical
regime, neglecting Coulomb interaction effects~\cite{Mat02,Ras13}. Different solutions for the 
Kuramoto model
where investigated semi-analytically in Ref.~\onlinecite{Til11} and classified according
to two integers, $q$ (the winding number mentioned above) and
$l$ (the number of angle differences larger than $\pi/2$) in Ref.~\onlinecite{Roy12}. 
In a somewhat different but related direction of investigation, 
Wiley et al.~\cite{Wil06} investigated the size of the basin of attraction for 
synchronous solutions with different $q$ in a cycle network of identical oscillators and found
that it gets smaller at higher $q$.

Below we show that two different solutions to Eq.~(\ref{eq:Ksteady})
on any network differ only by loop currents. This provides additional motivation for
investigating loop currents as it rigorously connects them to multiple stable solutions to 
Eq.~\eqref{eq:Ksteady}. We thus investigate single-loop networks and show that 
the algebraic upper bound of Rogge and Aeyels~\cite{Rog04} and Ochab and G\'ora~\cite{Och10} 
is generically valid for the Kuramoto model on a nonoriented cycle with nearest-neighbor interactions.
We furthermore demonstrate that, for such networks,  at most one angle difference can exceed $\pi/2$.
Stable solutions are in particular restricted to only $l=0$ or $l=1$ in the classification scheme of Ref.~\onlinecite{Roy12}. 
We show that the number of stable solutions decreases
monotonically as the coupling strength is reduced, and that 
solutions with $l=1$ emerge continuously at lower coupling from solutions
with $l=0$. 

\paragraph*{}
The manuscript is organized as follows: Section~\ref{preliminaries} states the initial concepts and defines the model considered. 
Loop flows and their link with multiple solutions to the power flow equations are discussed in Section~\ref{loopWinding}.
Section~\ref{multipleSolutions} gives a complete study of the multiple stable solutions to the power flow equations on a cycle network.
Conclusions are given in Section~\ref{conclusion}.

\section{Definitions and fundamental concepts}\label{preliminaries}
We are interested in a class of problems represented by at least three important physical systems. 
We have already defined the Kuramoto model in Eq.~(\ref{eq:Kuramoto}), for which more details can be found in 
review articles~\cite{Str00,Ace05,Dor14}. We have briefly mentioned vortices in superfluids and
superconductors, as well as Josephson junction arrays, where circulating supercurrents are given by laws similar to 
Eq.~\eqref{eq:Ksteady} and for which a rather vast literature, including review articles, also 
exists~\cite{And66,Bla94,Faz01}.
These problems are well documented in the physics literature and we
therefore do not discuss them further.
Electric power grids are less known in physics and we start with a brief introduction to this third
class of problems, 
emphasizing its connection with the Kuramoto model of Eq.~(\ref{eq:Kuramoto}).

\subsection{Power Flow and Swing Equations}\label{ssec:powerflow}

Power grids are AC electric networks. They can be modeled as graphs with  
$n$ nodes where each node $i=1,...,n$ injects (consumes) a power 
$P_i>0$ ($P_i<0$). The edges of the graph represent 
electrical lines with a complex admittance $Y=G+i B$. 
Power grids span different voltage levels separated by transformers which, to a good
approximation, conserve power but neither current nor voltage. Additionally, the control variables 
are the injected and consumed powers, therefore the equations governing the behavior
of the system are expressed in terms of electric powers and not currents. Considering a generating
power plant, the balance between the source (mechanical, thermal, chemical or nuclear) power,
the transmitted (electric) power and the losses leads to 
the {\it swing equations}~\cite{Ber00}
\begin{equation}\label{Swing}
\dot{\theta_i}=P_i-\sum_{j=1}^n|V_i| \  |V_j| \, \left[ G_{ij} \cos(\theta_i - \theta_j)
+ B_{ij} \sin(\theta_i - \theta_j) \right ] \, , \qquad   i=1,\ldots n \, ,
\end{equation}
where $\theta_i$ is the angle between the currents $I_i = |I_i| \exp(i \omega t)$
and the voltages $V_i = |V_i| \exp(i \omega t+i \theta_i)$ (in a frame rotating
with the frequency $\omega/2 \pi = 50$ or 60 Hz of the grid), 
$\theta_i-\theta_j$ is taken in 
$\big(-\pi,\pi\big]$ and $G$ and $B$ are the conductance and susceptance matrices 
respectively~\cite{Ber00}. 
In Eq.~(\ref{Swing}), we already consider a simplified version of the 
swing equations, where we neglected the inertia of the (rotating) generators. We did that since our main
interest is to determine whether a solution is stable or not,
which is not influenced by the presence of an inertia term (note that the inertia influences stability
time scales~\cite{Ber00}). 
In most of our discussion we make a second approximation and
consider networks of purely susceptive lines with $G_{ij}=0$. This is a leading order approximation 
in the small parameter $G_{ij}/B_{ij} < 0.1$ valid for very
high voltage networks. With this approximation,
lines have no ohmic losses and all nodes are at the same voltage. For the sake of
simplicity, we will also consider lines with identical capacities and set 
$K\coloneqq |V_i||V_j|B_{ij}$, on all edges $\langle ij\rangle$. With all these approximations, Eq.~(\ref{Swing}) 
leads to an equation similar to Eq.~(\ref{eq:Kuramoto}), 
\begin{equation}\label{SwingSimple}
\dot{\theta_i}=P_i-K \sum_{j\sim i} \sin(\theta_i-\theta_j) \, , \qquad  i=1,\ldots n\, , 
\end{equation}
where the sum is taken over nodes $j$ connected to node $i$ ($j\sim i$) and phase-locked solutions are 
governed by the power flow equations, which reduce to the form of Eq.~(\ref{eq:Ksteady}),
\begin{equation}\label{PFsimple}
 P_i=K  \sum_{j\sim i}\sin(\theta_i-\theta_j)\, , \qquad  i=1,\ldots n\, .
\end{equation}
Electric power grids are balanced in steady-state, meaning that power injections exactly
compensate power consumptions, i.e. 
\[
 \sum_iP_i=0\, .
\]
Additionally, injected and consumed powers are confined to a compact support,
$P_i \in [P_{\rm min},P_{\rm max}]$ which is necessary for the 
existence of phase-locked synchronous solutions~\cite{Erm85,Hem93}.

We note finally that the quantity $P_{ij}\coloneqq K\sin(\theta_i-\theta_j)$ represents 
the power flow along line $\langle ij\rangle$, from site $i$ to site $j$, so that 
Eq.~(\ref{PFsimple}) can be rewritten as 
\begin{align}\label{kirchhoff}
 P_i = \sum_{j \sim i} P_{ij}\, , 
\end{align}
which is Kirchhoff's currents law. Below we often use
the power flow terminology and in particular we discuss {\it loop flows} to describe circulating flows 
around closed cycles that do not distribute power to consuming nodes.
Eq.~(\ref{PFsimple}) only depends on angle differences, thus any solution is 
defined up to an homogeneous displacement of all angles. This gauge invariance allows to 
arbitrarily define a reference node whose angle is set to zero. All other angles are then determined 
with respect to that reference angle.

\subsection{Stability}\label{stability}
The swing equations, Eq.~(\ref{SwingSimple}), 
govern the system's dynamics and allow to determine the
linear stability of solutions 
of Eq.~(\ref{PFsimple}).
Under small perturbations about such a phase-locked solution, $\theta_i^{(0)} \rightarrow \theta_i^{(0)}
+ \delta \theta_i$, the linearized dynamics reads
\begin{equation}\label{SwingLin}
\delta \dot{\theta_i}=-\sum_{j\sim i}K\cos(\theta_i^{(0)}-\theta_j^{(0)}) 
(\delta \theta_i-\delta \theta_j) \, , \qquad  i=1,\ldots n\, .
\end{equation}
The linear stability of the solution $\{\theta_i^{(0)}\}$  is therefore determined 
by the spectrum of the \emph{stability matrix} $M(\{\theta_i^{(0)}\})$, 
\begin{align}\label{stabilityM}
 M_{ij}\coloneqq\left\{
 \begin{array}{ll}
  K\cos(\theta_i^{(0)}-\theta_j^{(0)})\, ,&\text{if }i\neq j\, ,\\
  \displaystyle -\sum_{k\sim i}K\cos(\theta_i^{(0)}-\theta_k^{(0)})\, ,&\text{if }i=j\, ,
  \end{array}
  \right.
\end{align}
which depends on the angles at the phase-locked solution. 
The eigenvalues of $M(\{ \theta_i^{(0)} \})$ are called Lyapunov exponents.
Because $\sum_j M_{ji}  = \sum_j M_{ij} = 0$, for all $i$, the constant vector is
an eigenvector of $M$ with eigenvalue $\lambda_1 = 0$. This follows from the above
mentioned gauge invariance, 
where only angle differences between oscillators matter.
Furthermore, as $M$ is real symmetric, all its eigenvalues are real. 
Thus the synchronous state is stable if $M(\{ \theta_i^{(0)} \})$ 
is negative semidefinite and unstable otherwise. 
In other words, the synchronous solution remains stable as long as the 
largest nonvanishing eigenvalue $\lambda_2$ of $M(\{ \theta_i^{(0)} \})$ remains negative.

To the best of our knowledge, it was first mentioned in Ref.~\onlinecite{Tav72} that as long as 
all angle differences are in $\big[-{\pi}/{2},{\pi}/{2}\big]$, 
Gershgorin's circle theorem~\cite{Hor85} guarantees that $M$ is negative semi-definite.
Then all Lyapunov exponents are non-positive, which implies that any solution 
of Eq.~(\ref{PFsimple}) with $\theta_i^{(0)}-\theta_j^{(0)} \in \big[-\pi/2,\pi/2\big]$ on each
of the graph's edges is linearly stable. 
The same theorem allows to conclude that if 
$|\theta_i^{(0)}-\theta_j^{(0)}| > \pi/2$ on all edges, the solution is linearly unstable.  
Recent works have investigated solutions with a single angle difference larger than $\pi/2$
in a Kuramoto model on a cycle network~\cite{Til11,Roy12}.
However, little is known analytically if some of the angle differences are smaller
and some are larger than $\pi/2$, except on cycle networks with unidirectional nearest-neighbor  
coupling~\cite{Rog04}.
Below we fill this gap and show that at most one angle difference is bigger than $\pi/2$ and that 
a stable solution with one angle difference exceeding $\pi/2$ comes from a solution at larger $K$ with all 
angle differences smaller than $\pi/2$.

\section{Loop flows and winding number}\label{loopWinding}
In this section we show that different solutions of Eq.~(\ref{PFsimple}) for any network 
differ only by circulating loop flows. This rigorous result, which appeared in slightly different form in Ref.~\onlinecite{Dor13}, sheds light on the common wisdom
that Eq.~(\ref{PFsimple}) may have multiple stable solutions for networks with closed 
cycles~\cite{Jan03,Kor72,Rog04,Wil06}. 
Before we discuss this theorem, we recall some definitions from graph theory which we will use.

\begin{defi}
A \textbf{graph} $G=(\mathcal{V}_G,E_G)$ is a set of vertices $\mathcal{V}_G$ with 
a set of edges $E_G$, each of which is a pair of connected vertices. If 
$i,j\in\mathcal{V}_G$, the edge connecting $i$ to $j$ is $\langle ij\rangle\in E_G$.
\end{defi}

\begin{defi}
 A \textbf{path} from vertex $i$ to vertex $j$ in a graph $G$ is a sequence $S\subset E_G$ of edges
 \begin{align*}
  S=\{\langle ii_1\rangle,\langle i_1i_2\rangle,...,\langle i_{\ell}j\rangle\}\, .
 \end{align*}
\end{defi}

\begin{defi}
 A graph is \textbf{connected} if for any two vertices $i,j\in\mathcal{V}_G$ there exists a path from $i$ to $j$.
\end{defi}

\begin{defi}
A \textbf{cycle} in a graph is a path from a vertex $i$ to itself going at most once through any edge. 
\end{defi}

\begin{defi}
A \textbf{tree} is a connected graph with no cycle. 
Given a graph $G=(\mathcal{V}_G,E_G)$, a \textbf{spanning tree} $T$ of $G$ is a tree such that $\mathcal{V}_T=\mathcal{V}_G$ and $E_T\subset E_G$.
\end{defi}

\begin{remarks}
 \begin{enumerate}[(i)]
  \item It can be shown inductively that a tree with $n$ vertices has exactly $n-1$ edges.
  \item On a tree-network, there is a unique flow distribution satisfying Kirchhoff's current law.
 \end{enumerate}
\end{remarks}

In what follows, we use the terms \emph{network} and \emph{grid} to denote physical objects, 
whose mathematical representations will be referred to as \emph{graphs}.
Additionally, we introduce the concept of \emph{loop flows}, which are
constant power flows circulating clockwise or anticlockwise around a cycle in a network. 
Strictly speaking, loop flows can be univocally defined only when power is neither
injected nor consumed in the network. With finite power injections and consumptions, loop flows can be defined only
relatively, as flow differences from a reference solution, in the spirit of the upcoming theorem.  

\paragraph*{}
Let $G$ be a graph and $O_G$ an arbitrary orientation of this graph, which means that we define positive and negative 
directions for every edge of $G$ in the following way. 
For each edge $\langle ij\rangle$ we call the vertex $i$ the \emph{source} of the edge and $j$ its \emph{target}.
Consider the real vector space $\mathcal{I}\simeq\mathbb{R}^m$ of flows on the $m$ edges of a
graph $G$. The components $\{I_\ell\}$
of a flow vector $\bm{I}\in\mathcal{I}$ describe the intensity of the flow on the $\ell^{\rm th}$
edge of $G$, with $I_\ell>0$ if the direction of the flow agrees with the orientation of this edge given by $O_G$, 
and $I_\ell<0$ otherwise. 
The canonical basis of $\mathcal{I}$ is the set of flow vectors 
$\bm{J}_\ell$, $\ell=1,...,m$, with unit flow on edge $\ell$ and zero flow on all other edges. 
Given a vector of power injections and consumptions at every node,
\begin{align*}
 \bm{P}&=(P_1,...,P_n)\in(1,...,1)^\perp\subset\mathbb{R}^n\, ,
\end{align*}
a flow vector $I\in\mathcal{I}$ satisfies Kirchhoff's current law
$P_i = \sum_{j \sim i} P_{ij}$ if 
\begin{align}\label{mlfs_balance2}
 P_i&=\sum_{\ell}A_{i\ell}I_\ell\qquad i=1,...,n\, ,
\end{align}
where we introduced the \emph{incidence matrix} $A$ of $G$,
\[
 A_{i\ell}\coloneqq\left\{
 \begin{array}{ll}
  1\, ,&\text{if node }i\text{ is the source of edge }\ell \, ,\\
  -1\, &\text{if node }i\text{ is the target of edge }\ell \, ,\\
  0\, ,&\text{otherwise}\, .
 \end{array}
 \right.
\]
We are now ready to formulate and prove our theorem.

\begin{thm}\label{thm_loop_flow}
 Let $G=(\mathcal{V}_G,E_G)$ be a connected graph with $\left|\mathcal{V}_G\right|=n$ sites and $\left|E_G\right|=m$ edges. Let $\bm{P}\in(1,...,1)^{\perp}$ be a vector of power 
 injections and consumptions at each node. Then two distributions of flows on $G$
 represented by flow vectors $\bm{I}'$ and $\bm{I}''\in\mathcal{I}$ satisfying Kirchhoff's currents law, Eq.~\eqref{kirchhoff},
 differ by a combination of loop flows on the different cycles of $G$.
\end{thm}

\begin{remark}
 In particular, Theorem~\ref{thm_loop_flow} implies that the flow distributions of two different solutions of Eq.~(\ref{PFsimple}) differ by a collection of loop flows. This result already appeared in slightly
 different form in the Supporting Information of Ref.~\onlinecite{Dor13}.
\end{remark}

\begin{proof}
If $m=n-1$, then $G$ is a tree and the flows on the lines are uniquely determined, which agrees with the statement because $G$ has no cycle. Therefore, from now on we assume $m\geq n$. Let 
$T$ be a spanning tree of $G$ and let us number the edges of $T$ from $1$ to $n-1$ and the edges of $G\setminus T$ from $n$ to $m$. Let $\bm{I}^{\circ}\coloneqq \bm{I}'-\bm{I}''$ be the 
difference between the two flow vectors. Then, for any $i$ we have
\begin{align*}
 \sum_{\ell}A_{i\ell}^\iStar I^{\circ}_{\ell}&=\sum_{\ell}A_{i\ell}^\iStar(I'_{\ell}-I''_{\ell})=P_i^\iStar-P_i^\iStar=0\, ,
\end{align*}
from which we conclude that $\bm{I}^{\circ}$ is a solution of Eq.~(\ref{mlfs_balance2}) with $\bm{P}=0$. What we need to show is therefore that any solution $\bm{I}$ of the system of equations
\begin{equation}\label{mlfs_balance_red2}
 \sum_{\ell}A_{i\ell}I_{\ell}=0\ ,\qquad i=1,...,n\, ,
\end{equation}
is a combination of loop flows. To do this we write Eq.~(\ref{mlfs_balance_red2}) in matricial form,
\begin{equation}\label{mlfs_balance_mat2}
 A\bm{I}=0\, .
\end{equation}
By definition, the set of solutions of Eq.~(\ref{mlfs_balance_mat2}) is the kernel of $A$, which is a subspace of $\mathcal{I}$. 

\paragraph*{}
In algebraic graph theory, $\ker(A)$ is referred to as the \emph{cycle space} and it is a standard result~\cite{Big93} that any element in $\ker(A)$ is a linear combination of unitary flows 
along the cycles of the network considered. This completes the proof.
\end{proof}

\begin{remark}
 Theorem~\ref{thm_loop_flow} is not restricted to the power flow problem. It generically applies 
 to any system of coupled oscillators with antisymmetric coupling, in particular to the Kuramoto model
 on any network.
\end{remark}

\paragraph*{}
Indexing the nodes along one such cycle, we write $P_{i,i+1}$ for the power flow from node $i$ to node $i+1$, with indices
taken modulo $n$. Theorem~\ref{thm_loop_flow} states in particular that multiple solutions to Eq.~(\ref{PFsimple}) can appear 
only when there are closed cycles in the network. 

\paragraph*{}
Alternatively, any flow $\{P_{i,i+1}\}$ on a cycle can be written as the sum of a reference solution, characterized by its flows $\{P_{i,i+1}^*\}$, 
and a loop flow of intensity $K\eps$, circulating around the cycle (see Fig.~\ref{cycle_model}),
\begin{align*}
 P_{i,i+1}^\iStar&=P_{i,i+1}^*+K\eps\, .
\end{align*}
We call $\eps\in\big[-1,1\big]$ the \emph{loop flow parameter}. 
It is only defined with respect to a reference solution, which is conveniently constructed from the $P_i$'s as
\begin{align*}
 P_{i,i+1}^*&\coloneqq\sum_{j=1}^iP_j\, ,\qquad i=1,...,n\, .
\end{align*}
Note that the reference solution depends on node numbering and any other flow distribution satisfying Kirchhoff's power balance can be taken as reference solution.
\begin{figure}[b]
 \begin{center}
  \includegraphics[width=175px]{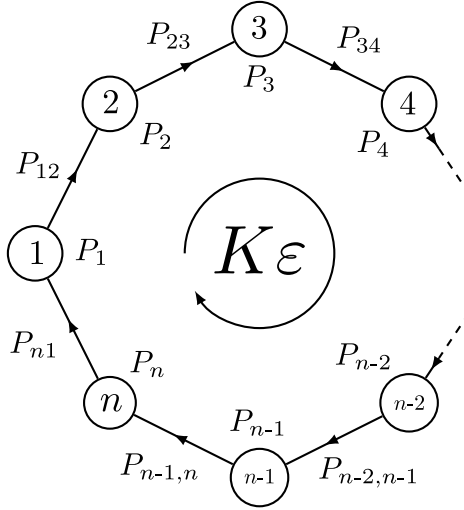}
  \caption{\textit{Cycle network of length $n$. The nodes inject/consume a power $P_i$, 
  while a power $P_{ij}$ is transmitted along the edge $\langle ij \rangle$. Edges
  correspond to lines with capacity $K$ (loop flow problem), to couplings of strength $K$ (Kuramoto problem) 
  or to tunnel barriers of transparency $\hbar K/8 e n_s$ (Josephson junction problem). 
  A loop flow of intensity $K\eps$ is circulating 
  around the cycle, where $\eps \in [-1,1]$.}}
  \label{cycle_model}
 \end{center}
\end{figure}

\paragraph*{}
As angle differences $\Delta_{ij}\coloneqq\theta_i-\theta_j$ are taken modulo $2\pi$ in the interval $(-\pi,\pi]$, 
their sum over the cycle gives an integer multiple of $2\pi$. This brings us to the definition of the winding number.

\begin{defi}
 For a given solution of the power flow Eq.~(\ref{PFsimple}) on the cycle, we define its \textbf{winding number} as the integer
\begin{align}\label{windingNumber}
 q\coloneqq\left(2\pi\right)^{-1}\sum_{i=1}^n\Delta_{i,i+1}\in\mathbb{Z}\, .
\end{align}
\end{defi}
The winding number characterizes a solution and is related to the loop flow. Unlike the latter, however, it is uniquely defined.
Eq.~(\ref{windingNumber}) quantizes the loop flow, i.e. it can take only discrete values.

\section{The number of stable solutions}\label{multipleSolutions}
Theorem~\ref{thm_loop_flow} connects the existence of multiple solutions to Eq.~(\ref{PFsimple}) to the presence of cycles traveled by 
quantized loop flows. The number of solutions is thus related to the number of acceptable, discrete loop flows. In the remainder of this manuscript, 
we focus on this problem in single-cycle graphs. 

\paragraph*{}
We first treat the case $K\to\infty$, where stable solutions necessarily have all angle differences in $[-\pi/2,\pi/2]$. 
We then consider the situation for finite $K$, where we show that the number of stable solutions decreases with $K$, that the angle difference along 
some of the lines can exceed $\pi/2$, but that it can happen on a single line at most.

\subsection{Angle differences and sum of angle differences}
A solution is fully characterized by the angle differences along the lines. 
These can be written as functions of the loop flow parameter $\eps$,
\begin{align}
 P_{i,i+1}^{\iStar}&=P_{i,i+1}^*+K\eps=K\sin(\Delta_{i,i+1})\quad\implies\quad\Delta_{i,i+1}=a_i(\eps)\, ,\label{angleDiff}
\end{align}
where there are two possible choices for each $a_i$,
\begin{align}\label{arcsine}
 a_i(\eps)&=\left\{
 \begin{array}{ll}
  \arcsin\left(\eps+{P_{i,i+1}^*}/{K}\right)&\implies\Delta_{i,i+1}\in\left[-{\pi}/{2},{\pi}/{2}\right]\, ,\\
  \pi-\arcsin\left(\eps+{P_{i,i+1}^*}/{K}\right)&\implies\Delta_{i,i+1}\in\left(-\pi,-{\pi}/{2}\right)\cup\left({\pi}/{2},\pi\right]\, .
 \end{array}
 \right.
\end{align}
Since the power transmitted along any link is bounded by $K$, we obtain bounds on $\eps$,
\begin{align*}
 -K\leq P_{i,i+1}^{\iStar}\leq K&\iff-1-{P_{i,i+1}^*}/{K}\leq\eps\leq1-{P_{i,i+1}^*}/{K}\, ,& i=1,...,n\, . 
\end{align*}
Thus $\eps\in\big[\eps_{\min},\eps_{\max}\big]$, with 
\begin{subequations}
\begin{align}
 \eps_{\min}&\coloneqq\max_{1\leq i\leq n}\left\{-1-{P_{i,i+1}^*}/{K}\right\}=-1-{P_{\min}^*}/{K}\, ,\label{eplus}\\
 \intertext{and}
 \eps_{\max}&\coloneqq\min_{1\leq i\leq n}\left\{1-{P_{i,i+1}^*}/{K}\right\}=1-{P_{\max}^*}/{K}\, ,\label{emoins}
\end{align}
\end{subequations}
where $P_{\min}^*\coloneqq\min_{i}\P{i}$ and $P_{\max}^*\coloneqq\max_iP_{i,i+1}^*$. 
Note that as soon as the $P_i$'s are not all equal to zero, $P_{\min}^*\neq P_{\max}^*$. 
We add an appropriate constant to the reference flow to make sure that $P_{\min}^*\neq0$ and $P_{\max}^*\neq0$\label{pmaxNeqZero}, 
which will facilitate our discussion without restricting its generality.

\paragraph*{}
As seen in Section~\ref{stability}, a solution is stable if all angle differences belong to
the interval $\big[-{\pi}/{2},{\pi}/{2}\big]$. In this situation, we can write the sum of angle differences around the cycle as a 
function of the parameter $\eps$,
\begin{align}\label{Afunction}
 \A_0(K,\eps)&\coloneqq\sum_{i=1}^n\Delta_{i,i+1}=\sum_{i=1}^n\arcsin\left(\eps+{\P{i}}/{K}\right)\, .
\end{align}
As the arcsine is continuous and increasing, the function $\A_0$ is also continuous and increasing with respect to $\eps$. Thus for fixed $K_0$, the function $\A_0(K_0,\eps)$ defines a one-to-one 
correspondence between the intervals 
\[
 \big[\eps_{\min}(K_0),\eps_{\max}(K_0)\big]\quad \longleftrightarrow\quad \big[\A_0(K_0,\eps_{\min}(K_0)),\A_0(K_0,\eps_{\max}(K_0))\big]\, .
\]
The sum of angle differences around the cycle has to be a multiple of $2\pi$, thus defining $\eps_q$ such that $\A_0(K_0,\eps_q)=2\pi q$ and 
\begin{align*}
 \Delta_{i,i+1}&=\arcsin\left(\eps_q+{\P{i}}/{K_0}\right)\, , \qquad i=1,...,N\, ,
\end{align*}
gives a stable solution of Eq.~(\ref{PFsimple}).

\paragraph*{}
Therefore, the number of solutions with $|\Delta_{i,i+1}|<\pi/2$ for all $i$ is straightforwardly given by the number of $q$'s such that $\A_0(K_0,\eps_q)=2\pi q$. 
Previous works have treated this case~\cite{Och10}, 
however allowing $|\Delta_{i,i+1}|>\pi/2$ renders the problem much more complicated. It has so far been solved only for unidirectional 
coupling~\cite{Rog04}. Our strategy for incorporating solutions with $|\Delta_{i,i+1}|>\pi/2$ is to first treat $K\to\infty$, where we show that $|\Delta_{i,i+1}|<\pi/2$, for all $i$, 
for stable solutions. The number of solutions is then easy to compute. 
Second, we generalize the study to finite $K$ and see that the number of solutions obtained for $K\to\infty$ is an upper bound on the number of solutions for any finite $K$.

\subsection{Infinite capacity}\label{infK}
The case $K\to\infty$ is equivalent to the identical oscillators case with $P_i=0$, for all $i$. In this limit, the bounds on $\eps$ are
\begin{align*}
 \lim_{K\to\infty}\eps_{\max}(K)&=1\, ,&\lim_{K\to\infty}\eps_{\min}(K)&=-1\, ,
\end{align*}
and thus 
\begin{align*}
 \lim_{K\to\infty}\A_0(K,\eps_{\min}(K))&=-{n\pi}/{2}\, ,&\lim_{K\to\infty}\A_0(K,\eps_{\max}(K))&={n\pi}/{2}\, .
\end{align*}
An $\eps_q$ is associated to each integer multiple of $2\pi$ in $\big[-{n\pi}/{2},{n\pi}/{2}\big]$ corresponding to a stable 
solution of Eq.~(\ref{PFsimple}). There are $\mathcal{N}=2\, {\rm Int}[n/4]+1$ such integers. 
This is illustrated in Fig.~\ref{Aplot_1cycle}.
\begin{figure}[b]
\begin{center}
 \includegraphics[width=225px]{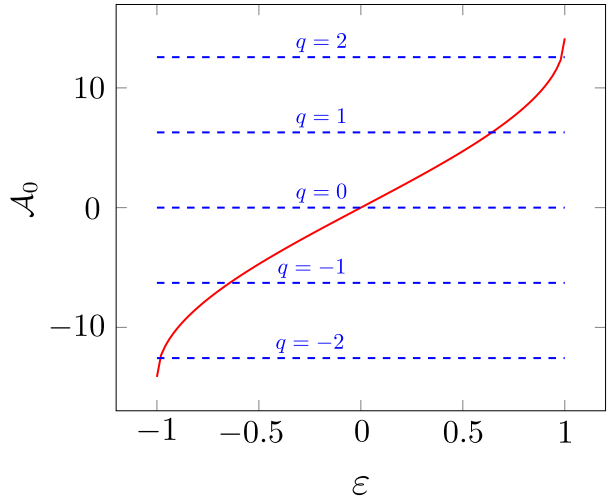}
 \caption{\textit{Plot of $\A_0(K,\eps)$ as a function of $\eps$ (red), for a cycle network of length $n=9$, with $K\to\infty$. 
 Horizontal dashed lines correspond to $\A_0=2\pi q$ with $q$ values indicated.
 Each intersection of the red line with a blue dashed line gives a stable solution of Eq.~(\protect\ref{PFsimple}).}}
 \label{Aplot_1cycle}
\end{center}
\end{figure}

\begin{thm}\label{largeK}
 For $K\to\infty$, any stable solution of the power flow Eq.~(\ref{PFsimple}) on a cycle network 
 has all angle differences in 
 $\big[-{\pi}/{2},{\pi}/{2}\big]$. Furthermore all angle differences are 
 equal to ${2\pi q}/{n}$, where $q$ is the winding number of the solution.
\end{thm}

\begin{remarks}
\begin{enumerate}[(i)]
 \item The \emph{principal minors} of a matrix $A$ are the determinants of the square submatrices of $A$ with the same row and column indices. 
 \emph{Sylvester's criterion} states that a matrix is positive semi-definite if and only if all its principal minors are non-negative~\cite{Hor85}.
 \item The result of Theorem~\ref{largeK} was already known for unidirectional coupling~\cite{Rog04}. 
 Here we extend this result to bidirectional interactions. Furthermore, our 
 approach allows to relate the finite $K$ situation to the infinite $K$ situation.
 \end{enumerate}
\end{remarks}

\begin{proof}
The power flow along line $\langle i,i+1\rangle$ is 
\[
 P_{i,i+1}^\iStar=K\sin(\Delta_{i,i+1})=P_{i,i+1}^*+K\eps\, ,
\]
where $\{P_{i,i+1}^*\}$ is a reference solution constructed from finite powers.
Thus when $K\to\infty$, the sine of the angle difference along every line of the cycle tends to the same value,
\[
 \lim_{K\to\infty}\sin(\Delta_{i,i+1})=\eps\, .
\]
This implies that the angle difference along each line of the network belongs to the set $\{\arcsin(\eps),\pi-\arcsin(\eps)\}$ 
and thus the cosine of the angle differences along all the lines takes the same absolute value with either positive or negative sign. 
First of all, if all angle differences are $\arcsin(\eps)$, the stability matrix defined in Eq.~\eqref{stabilityM} is easily expressed as
\[
 M=Kc
 \begin{pmatrix}
  -2 & 1 & & & 1\\
  1 & -2 & \ddots & &\\
  & \ddots & \ddots & & 1\\
  1 & & & 1 & -2
 \end{pmatrix}\, ,
\]
where $c\coloneqq\cos(\arcsin(\eps))=\sqrt{1-\eps^2}$. This matrix is negative semi-definite and has only non-positive eigenvalues by 
Gershgorin's circle theorem~\cite{Hor85}. Thus the solution is stable. 
Now, if all angle differences are $\pi-\arcsin(\eps)$, then all cosines are negative and the stability matrix is obviously 
positive semi-definite. The solution is then unstable. 
Let us now consider the mixed case where at least one angle difference is $\arcsin(\eps)$ and one is $\pi-\arcsin(\eps)$. 
In this case, there exists at least one node $i$ such that $\Delta_{i-1,i}=\arcsin(\eps)$ and $\Delta_{i,i+1}=\pi-\arcsin(\eps)$
and the corresponding stability matrix has the form
\begin{equation}\label{submatrix}
 M'=Kc
 \begin{pmatrix}
  \ddots & \ddots & & & \\
  \ddots & x & 1 & 0 & \\
  & 1 & 0 & -1 & \\
  & 0 & -1 & y & \ddots\\
  & & & \ddots & \ddots
 \end{pmatrix}\, .
\end{equation}
The principal minor of $-M'$ with row and column indices $\{i,i+1\}$ is
\[
 \begin{vmatrix}
  0 & +1\\
  +1 & -y
 \end{vmatrix}=-1\, ,
\]
which, by Sylvester's criterion~\cite{Hor85}, implies that $M'$ is not negative semi-definite. 
In other words, $M'$ has at least one positive eigenvalue and thus the solution 
is unstable. From this we conclude that the stable solutions for sufficiently large $K$ 
all have angle differences in $\big[-{\pi}/{2},{\pi}/{2}\big]$. They are captured by finding the intersections of $\A_0(K,\eps)$
with integer multiples of $2\pi$ as illustrated in Fig.~\ref{Aplot_1cycle}.

\paragraph*{}
Let $q$ be the winding number of a stable solution for $K\to\infty$. As all angle differences have the same value $\Delta\in\big[-{\pi}/{2},{\pi}/{2}\big]$, we have
\begin{align*}
 2\pi q=\sum_{i=1}^n\Delta_{i,i+1}=n\Delta&\implies\Delta={2\pi q}/{n}\, .
\end{align*}
The corresponding angles are $\theta_i=-{2\pi q i}/{n}$, taken in the interval $\big(-\pi,\pi\big]$.
\end{proof}

\subsection{Finite capacity}\label{finK}
We now consider finite values for $K$ and $P_i$'s not all equal to zero.
We first show that the number of solutions to the power flow Eq.~(\ref{PFsimple}) with all angle differences 
in $\big[-\pi/2,\pi/2\big]$ decreases with $K$. Second, we show that for finite $K$, there exist 
stable solutions with one angle difference in $\big(-\pi,-\pi/2\big)\cup\big(\pi/2,\pi\big]$, 
and we relate them to solutions at larger $K$ with all angle differences in $\big[-\pi/2,\pi/2\big]$. 
This gives an analytical confirmation of the numerical observations of Tilles et al.~\cite{Til11}, 
and of Roy and Lahiri~\cite{Roy12}.

\begin{prop}\label{prop:A0(K,eps_min) monotonous with K}
 For a one-cycle network with $n$ nodes, if $K$ decreases, then $\A_0(K,\eps_{\min})$ increases and $\A_0(K,\eps_{\max})$ decreases. 
\end{prop}

\begin{proof}
From Eqs.~(\ref{emoins}) and (\ref{Afunction}) the derivative of $\A_0$ with respect to $K$ reads
\begin{align*}
 \frac{d\A_0(K,\eps_{\min})}{dK}&={\sum}'\left[{1-\left(-1+\frac{\P{i}-P_{\min}^*}{K}\right)^2}\right]^{-\frac{1}{2}}\frac{P_{\min}^*-\P{i}}{K^2}\, ,
\end{align*}
where $\Sigma'$ indicates that the sum is taken over indices $j$ such that $P_{j,j+1}^*>P_{\min}^*$. This sum is obviously non-positive. In the same way it is easily seen that
\[
 \frac{d\A_0(K,\eps_{\max})}{dK}\geq0\, .
\]
\end{proof}

\paragraph*{}
Proposition~\ref{prop:A0(K,eps_min) monotonous with K} implies that as $K$ decreases, the interval of values of $\A_0$ 
gets smaller and contains fewer and fewer multiples of $2\pi$. We show now that for finite capacities, the stable solutions are directly 
related to the stable solutions for $K\to\infty$, even if some of them have angle differences in $\big(-\pi,-{\pi}/{2}\big)\cup\big({\pi}/{2},\pi\big]$. First we define 
\begin{align*}
 \A_j(K,\eps)&\coloneqq\sum_{i\neq j}\arcsin\left(\eps+{\P{i}}/{K}\right)+\pi-\arcsin\left(\eps+{\P{j}}/{K}\right)\, ,\qquad j=1,...,n\, .
\end{align*}
The function $\A_0$ is the sum of angle differences all taken in the interval $\big[-{\pi}/{2},{\pi}/{2}\big]$ and 
for $j=1,...,n$, the function $\A_j$ is this sum 
with one angle difference, the $j^{\text{th}}$, taken in $\big(-\pi,-{\pi}/{2}\big)\cup\big({\pi}/{2},\pi\big]$. 
We also introduce the following notation
\begin{align*}
 c_i&\coloneqq\cos(\Delta_{i,i+1})\, .
\end{align*}
The sign of $c_i$ depends on our choice for $\Delta_{i,i+1}$,
\begin{align*}
 c_i&=\left\{
 \begin{array}{ccc}
  \cos\left[\arcsin\left(\eps+{\P{i}}/{K}\right)\right]&=&\sqrt{1-\left(\eps+{\P{i}}/{K}\right)^2}\, ,\\
  \cos\left[\pi-\arcsin\left(\eps+{\P{i}}/{K}\right)\right]&=&-\sqrt{1-\left(\eps+{\P{i}}/{K}\right)^2}\, .
 \end{array}
 \right.
\end{align*}
The domain $\mathcal{D}$, in the $(K,\eps)$-plane, where the functions $\A_j$ are defined, is such that each arcsine 
is well-defined,
\begin{align}\label{eq:D}
 \mathcal{D}=\left\{(K,\eps)\colon\eps+{P_{\min}^*}/{K}\geq-1,\eps+{P_{\max}^*}/{K}\leq1\right\}\, .
\end{align}
By definition, in the interior of $\D$, the $c_i$'s are nonzero. Let us define the upper and lower boundaries of $\D$,
\begin{align}\label{eq:D01}
\begin{split}
 \D_1&\coloneqq\left\{(K,\eps)\colon\eps+{P_{\max}^*}/{K}=1\right\}\, ,\\
 \D_0&\coloneqq\left\{(K,\eps)\colon\eps+{P_{\min}^*}/{K}=-1\right\}\, .
\end{split}
\end{align}
We next denote by $j_0$ and $j_1$ the indices such that $P_{j_0,j_0+1}^*=P_{\min}^*$ and $P_{j_1,j_1+1}^*=P_{\max}^*$ respectively. 
Note that for $(K,\eps)\in\D_1$ [resp. $(K,\eps)\in\D_0$], the functions $\A_0(K,\eps)$ and $\A_{j_1}(K,\eps)$ [resp. $\A_{j_0}(K,\eps)$] have the same value.

\begin{remark}
It is possible that multiple lines carry the same maximal or minimal power. In this case these indices are not uniquely 
defined, but we are free to choose any $j_0$ and $j_1$ satisfying $P_{j_0,j_0+1}^*=P_{\min}^*$ and $P_{j_1,j_1+1}^*=P_{\max}^*$.
\end{remark}

\paragraph*{}
For any choice of $a_i$'s in Eq.~(\ref{arcsine}), any point $(K,\eps)\in\D$ such that $\sum_ia_i=2\pi q$ is a solution (not necessarily stable) of Eq.~(\ref{PFsimple}). 
Hence we now study the $2\pi q$-level sets of $\A_j$, for $q\in\mathbb{Z}$ and $j=0,...,n$. 
Note first that as $\A_j$ is smooth in the interior of the domain $\D$ for any $j$, the Implicit Function Theorem~\cite{Che01} implies that its level sets are level curves.
For any $K_0$, we call $\mathcal{S}(K_0)\subset\big[\eps_{\min},\eps_{\max}\big]$ the set of $\eps$-values 
corresponding to stable solutions of the power flow Eq.~(\ref{PFsimple}), 
i.e. such that there exists a choice of $\{a_i\}$ in Eq.~(\ref{arcsine}) for which
\begin{align*}
 \sum_i a_i(K_0,\eps)&=2\pi q,\qquad q\in\mathbb{Z}\, .
\end{align*}
Let $\mathcal{N}(K_0)\coloneqq|\mathcal{S}(K_0)|$ be its cardinality. 
The main results of this section are the following theorem on the properties of $\mathcal{N}(K)$ and its corollary.

\begin{thm}\label{mainThm}
 The number of stable solutions of the power flow equations, $\mathcal{N}(K)$ is a monotonically increasing function of $K$.
\end{thm}

\begin{cor}\label{mainCor}
 The value 
 \begin{align*}
  \mathcal{N}_{\infty}&\coloneqq\lim_{K\to\infty}\mathcal{N}(K)=2\, {\rm Int}[n/4]+1\, ,
 \end{align*}
 is an upper bound on the number of stable solutions of Eq.~(\ref{PFsimple}) on a cycle network, independently of $K$ and $\{P_i\}$.
\end{cor}

\paragraph*{}
The proof of Theorem~\ref{mainThm} relies on five lemmas. A major ingredient of the proof is that the functions $\A_j(K,\eps)$, for $j=1,...,n$, have no critical points. This fact and 
Lemma~\ref{finally} give precise informations about the shape of the level curves of $\A_j$.

\begin{lem}\label{blabla2}
 For $j\in\{1,...,n\}$, the function $\A_j$ has no critical point in the interior of $\D$.
\end{lem}

\begin{proof}
For $j\in\{1,...,n\}$, we have
\begin{align*}
 \frac{\partial\A_j}{\partial\eps}&=\sum_kc_k^{-1}\, .
\end{align*}
Assume first that at some point $(K,\eps)$ in the interior of $\D$, ${\partial\A_j}/{\partial\eps}=0$, then
\begin{align}\label{dAdepseg0}
 \frac{\partial\A_j}{\partial\eps}=0&\iff\sum_kc_k^{-1}=0\iff\sum_{k\neq j}c_k^{-1}=-c_j^{-1}\iff\sum_{k\neq j}-{c_j}/{c_k}=1\, .
\end{align}
Recall that as we chose 
\begin{align*}
 \Delta_{j,j+1}=\pi-\arcsin\left(\eps+{P_{j,j+1}^*}/{K}\right)\, ,
\end{align*}
we have $c_j<0$. It is then easy to check that for any $k\neq j$
\begin{align*}
 0<-{c_j}/{c_k}<1&\implies0<-c_j<c_k\\
 &\implies \sqrt{1-\left(\eps+{\P{j}}/{K}\right)^2}<\sqrt{1-\left(\eps+{\P{k}}/{K}\right)^2}\\
 &\implies \left(\eps+{\P{j}}/{K}\right)^2>\left(\eps+{\P{k}}/{K}\right)^2\, .
\end{align*}
There are now two possible cases :
\begin{enumerate}
 \item if $\eps+{\P{j}}/{K}>0$, then
 \begin{align*}
  \eps+{\P{j}}/{K}>\eps+{\P{k}}/{K}&\iff \P{j}>\P{k}\, ,\qquad\forall k\neq j\\
  &\implies \P{j}=P_{\max}^*\, ;
 \end{align*}
 \item if $\eps+{\P{j}}/{K}<0$, then
 \begin{align*}
  \eps+{\P{j}}/{K}<\eps+{\P{k}}/{K}&\iff \P{j}<\P{k}\, ,\qquad\forall k\neq j\\
  &\implies \P{j}=P_{\min}^*\, .
 \end{align*}
\end{enumerate}
Thus if $j\notin\{j_0,j_1\}$, Eq.~(\ref{dAdepseg0}) cannot hold and $\A_j$ has no critical point in $\D$. Let now $j\in\{j_0,j_1\}$ and assume that 
\begin{align*}
 \frac{\partial\A_j}{\partial\eps}&=\sum_kc_k^{-1}=0\, .
\end{align*}
We calculate
\begin{align*}
 \frac{\partial\A_j}{\partial K}&=\sum_{k\neq j}\frac{\partial}{\partial K}\arcsin\left(\eps+{\P{k}}/{K}\right)
 +\frac{\partial}{\partial K}\left[\pi-\arcsin\left(\eps+{\P{j}}/{K}\right)\right]\\
 &=-\sum_{k\neq j}c_k^{-1}{\P{k}}/{K^2}-c_j^{-1}{\P{j}}/{K^2}\\
 &=-\sum_{k\neq j}c_k^{-1}{\P{k}}/{K^2}+\sum_{k\neq j}c_k^{-1}{\P{j}}/{K^2}\\
 &=\sum_{k\neq j}\left({\P{j}-\P{k}}\right)/\left({K^2c_k}\right)\, ,
\end{align*}
which is non zero as every term is non-negative (resp. non-positive) if $j=j_1$ (resp. $j=j_0$). 
Thus the partial derivatives of $\A_j$ are never simultaneously zero implying that 
$\A_j$ has no critical point in the domain $\D$. 
\end{proof}

\begin{cor}\label{shape}
 For any $j\in\{1,...,n\}$, the level sets of $\A_j$ are continuous lines that: i) cannot end in the interior of $\D$, ii)
 are not closed and iii) have no trifurcation.
\end{cor}

\begin{proof}
Any of these situations would imply at least one critical point.
\end{proof}

\begin{lem}\label{finally}
 Let $L\in\mathbb{R}$. If there exist $K_0\in\mathbb{R}$ such that $\A_{j_1}(K_0,\eps_{\max}(K_0))=L$, then there is a single level curve of $\A_{j_1}=L$ starting at 
 $(K_0,\eps_{\max}(K_0))$. The same holds for level curves of $\A_{j_0}=L$ starting at $(K_0,\eps_{\min}(K_0))$.
\end{lem}

\begin{remark}
 This lemma means that the red curve in Fig.~\ref{intervals} is unique.
\end{remark}
\begin{figure}[b]
 \begin{center}
  \includegraphics[width=250px]{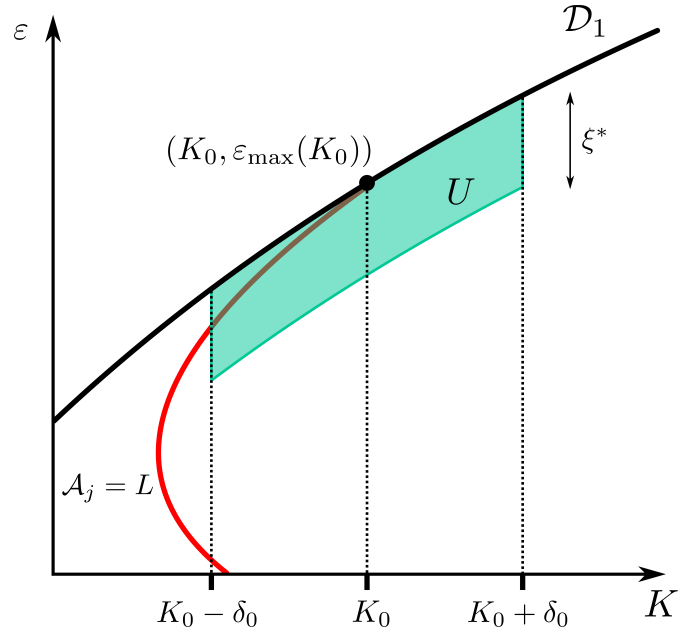}
  \caption{\textit{Sketch of the level curve $\A_j=0$ (red curve) and of the neighborhood $U$ of $(K_0,\eps_{\max}(K_0))$ 
  where the partial derivative of $\A_j$ with respect to $K$ always has the same non-zero sign. The neighborhood $U$ is bounded above 
  by the boundary of the domain $\D$, left and right by the bounds of the interval $I_0$ and is of height $\xi^*$.}}
  \label{intervals}
 \end{center}
\end{figure}

\begin{proof}
We prove the first statement, the proof of the second one being similar. We first recall that
\begin{align}\label{partialK}
 \frac{\partial\A_{j_1}}{\partial K}&=\left[1-\left(\eps+\frac{P_{\max}^*}{K}\right)^2\right]^{-\frac{1}{2}}\frac{P_{\max}^*}{K^2}
 -\sum_{k\neq j_1}\left[1-\left(\eps+\frac{\P{k}}{K}\right)^2\right]^{-\frac{1}{2}}\frac{\P{k}}{K^2}\, .
\end{align}
Consider now a small interval around $K_0$, $I_0=\big[K_0-\delta_0,K_0+\delta_0\big]$. 
For any $K\in I_0$, there is a $\xi_K>0$ such that for all 
$\eps\in\big(\eps_{\max}(K)-\xi_K,\eps_{\max}(K)\big)$, the first term in the right-hand-side of Eq.~(\ref{partialK}) dominates, and the partial derivative 
${\partial\A_{j_1}}/{\partial K}$ has the same sign as $P_{\max}^*$, which we chose non-zero above (see discussion below Eq.~\eqref{emoins}). Thus setting 
\begin{align*}
 \xi^*&\coloneqq\min_{K\in I_0}\{\xi_K\}\, ,
\end{align*}
allows us to define a neighborhood $U\subset\D$ of $(K_0,\eps_{\max}(K_0))$,
\begin{align*}
 U&\coloneqq\left\{(K,\eps)\colon |K-K_0|<\delta_0~,~0\leq\eps_{\max}(K)-\eps<\xi^*\right\}\, ,
\end{align*}
where the partial derivative ${\partial\A_{j_1}}/{\partial K}$ always has the same non-zero sign. 
This neighborhood is sketched in Fig.~\ref{intervals}. But if there is more than one level curve starting at 
$(K_0,\eps_{\max}(K_0))$ and since $\A_{j_1}$ is not constant, this partial derivative has to change sign in any neighborhood of $(K_0,\eps_{\max}(K_0))$, 
which leads to a contradiction. There is therefore at most one such level curve.
\end{proof}

\paragraph*{}
We next investigate how the linear stability of the solutions to the power flow Eq.~(\ref{PFsimple}) varies along the level curves $\A_j=2\pi q$. 
The two following lemmas show that the only functions leading to stable solutions are $\A_0$, $\A_{j_0}$ and $\A_{j_1}$.

\begin{lem}\label{bla}
 For any choice of $a_i$'s, the stability matrix $M$ has a second null eigenvalue $\lambda_2$ (see Section~\ref{stability}) if and only if
 \begin{align*}
  \sum_kc_k^{-1}=0\, .
 \end{align*}
\end{lem}

\begin{proof}
Consider the charateristic polynomial of the stability matrix $M$,
\begin{align*}
 \chi(M)&=\left|
 \begin{matrix}
  -c_1-c_n-\lambda&c_1&\cdots&0&\cdots&c_n\\
  c_1&-c_1-c_2-\lambda&c_2&&&\vdots\\
  \vdots&c_2&\ddots&\ddots&&0\\
  0&&\ddots&&&\vdots\\
  \vdots&&&&&c_{n-1}\\
  c_n&\cdots&0&\cdots&c_{n-1}&-c_{n-1}-c_n-\lambda
 \end{matrix}
 \right|\, .
\end{align*}
Adding all rows to the first one it can be written $\chi(M)=\det\left(\tilde{M}(\lambda)\right)$ with
\begin{align*}
 \tilde{M}(\lambda)&=
 \begin{pmatrix}
  -\lambda&\cdots&&&\cdots&-\lambda\\
  c_1&-c_1-c_2-\lambda&c_2&&&\vdots\\
  \vdots&c_2&\ddots&\ddots&&0\\
  0&&\ddots&&&\vdots\\
  \vdots&&&&&c_{n-1}\\
  c_n&\cdots&0&\cdots&c_{n-1}&-c_{n-1}-c_n-\lambda
 \end{pmatrix}\, .
\end{align*}
Expanding the determinant along the first row we obtain
\begin{align}
\begin{split}\label{sum_minors}
 \chi(M)&=(-\lambda)\sum_{i=1}^n(-1)^{i-1}\det\left([\tilde{M}(\lambda)]_{1i}\right)\\
 &\eqqcolon(-\lambda) Q(\lambda)\, ,
\end{split}
\end{align}
where $\det\left(\left[A\right]_{ij}\right)$ stands for the $(i,j)$-cofactor of $A$. One eigenvalue obviously vanishes 
and a second eigenvalue, $\lambda_2$, is zero if and only if $Q(0)=0$.

\paragraph*{}
We show now that for $i=2,...,n$,
\begin{align*}
 \det\left([\tilde{M}(0)]_{1,i}\right)&=-\det\left([\tilde{M}(0)]_{1,i-1}\right)\, .
\end{align*}
Let $\C_k$ denote the $k^{\text{th}}$ column of matrix $\tilde{M}(0)$ with the first row removed. We write
\begin{align*}
 \det\left([\tilde{M}(0)]_{1i}\right)&=\Bigg|\C_1~\cdots~\C_{i-2}\quad\C_{i-1}\quad\C_{i+1}~\cdots~\C_n\Bigg|\\
 &=\left|\C_1\quad\cdots\quad\C_{i-2}\quad\sum_{j\neq i}\C_j\quad\C_{i+1}\quad\cdots\quad\C_n\right|\\
 &=\Bigg|\C_1~\cdots~\C_{i-2}\quad-\C_i\quad\C_{i+1}~\cdots~\C_n\Bigg|\\
 &=-\det\left([\tilde{M}(0)]_{1,i-1}\right)\, ,
\end{align*}
where at the second line we used that the determinant is not changed by adding a linear combination of columns to any column, and at the third line, we used the fact that 
the sum of the elements of any row is zero. We conclude that
\begin{align}\label{minor2minor}
 \det\left([\tilde{M}(0)]_{1i}\right)=(-1)^{i-1}\det\left([\tilde{M}(0)]_{11}\right)\, .
\end{align}

\paragraph*{}
Thus to calculate $Q(0)$ we only have to compute $\det\left([\tilde{M}(0)]_{11}\right)$. 
Since $[\tilde{M}(0)]_{11}$ is tridiagonal we compute its $LU$-factorization using Thomas algorithm~\cite{Gan14},
\begin{align*}
 [\tilde{M}(0)]_{11}&=\left(
 \begin{matrix}
  -c_1-c_2&c_2&&\\
  c_2&-c_2-c_3&\ddots&\\
  &\ddots&\ddots&c_{n-1}\\
  &&c_{n-1}&-c_{n-1}-c_{n}
 \end{matrix}
 \right)\\ 
 &=
 \begin{pmatrix}
  \phantom{c_2/}1&&0&\\
  c_2/\beta_1&\phantom{c_2/\beta_2}1&&\\
  &\ddots&\ddots&\\
  \phantom{c_2/}0&&\displaystyle c_{n-1}/\beta_{n-2}&1
 \end{pmatrix}
 \cdot
 \begin{pmatrix}
  \beta_{1\phantom{-1}}&c_{2\phantom{-1}}&\phantom{\ddots}0&\\
  &\beta_{2\phantom{-1}}&\ddots\quad&\\
  &&\ddots\quad&c_{n-1}\\
  &0\phantom{\ddots}&&\beta_{n-1}
 \end{pmatrix}\, ,\\
 \intertext{where}
 \beta_i&\coloneqq\left\{
 \begin{array}{ll}
  -(c_1+c_2)\, ,&\text{if }i=1\ ,\\
  -(c_i+c_{i+1}+{c_i^2}/{\beta_{i-1}})\, ,&\text{if }i\neq1\, .
 \end{array}
 \right.
\end{align*}
This factorization is only valid for non-singular matrices, but by continuity it can be computed arbitrarily close to points where the determinant vanishes. 
Computing the determinant of the matrix $[\tilde{M}(0)]_{11}$ then reduces to computing the product of the $\beta_i$'s. 

\paragraph*{}
Let us define
\begin{align}\label{eq:Definition mu}
 \mu_i&\coloneqq\left\{
 \begin{array}{ll}
  1\, ,&\text{if }i=0\, ,\\
  \displaystyle\sumcos{i+1}\, ,&\text{if }i=1,...,n\, .
 \end{array}\right.
\end{align}
In Appendix~\ref{induction}, we prove by induction that $\mu_{i-1}\cdot\beta_i=-\mu_i$, for all $i=1,...,n$. 
This allows to compute the determinant of $[\tilde{M}(0)]_{11}$,
\begin{align}
 \det\left([\tilde{M}(0)]_{11}\right)&=\prod_{i=1}^{n-1}\beta_i
 =(-1)^{n-1}\frac{\mu_{n-1}}{\mu_0}
 =(-1)^{n-1}\sumcos{n}
 =(-1)^{n-1}\prod_{j=1}^nc_j^\iPower\cdot\sum_{k=1}^nc_k^{-1}\, ,\label{det_minor}
\end{align}
where the last equality holds as long as all $c_k$'s are nonzero, which is true in the interior of $\D$.

\paragraph*{}
Finally, combining Eqs.~(\ref{sum_minors}), (\ref{minor2minor}) and (\ref{det_minor}) we have
\begin{align}
 Q(0)&=(-1)^{n-1}n\prod_{j=1}^nc_j\sum_{k=1}^nc_k^{-1}\,,\nonumber\\
\intertext{and thus}
 \lambda_2=0&\iff Q(0)=0\iff\sum_{k=1}^nc_k^{-1}=0\, .\label{l2IsZ}
\end{align}
\end{proof}

\begin{lem}\label{blabla}
 For a given value $L\in\mathbb{R}$ and $j\in\{j_0,j_1\}$, there is at most one point where $\sum_kc_k^{-1}=0$ along a connected 
 component of the level curve $\A_j=L$.
\end{lem}

\begin{remark}
From Lemma~\ref{blabla2} we already know that if $j\notin\{j_0,j_1\}$, $\sum_kc_k^{-1}$ is never zero.
\end{remark}

\begin{proof}
We already know that
\begin{align*}
 \sum_kc_k^{-1}=\frac{\partial\A_j}{\partial\eps}\, .
\end{align*}
Hence the sum $\sum_kc_k^{-1}$ equals zero if and only if the level curve of $\A_j$ is parallel to the $\eps$ axis. 
Let us now differentiate the sum $\sum_kc_k^{-1}$ with respect to $\eps$ at 
such a point, to see how it varies along the level curve of $\A_j$. 
Using the fact that $c_j^{-1}=-\sum_{k\neq j}c_k^{-1}$ we have
\begin{align*}
 \begin{array}{rcl}
 \displaystyle\frac{\partial}{\partial\eps}\sum_{k=1}^nc_k^{-1}&=&\displaystyle\sum_{k\neq j}\displaystyle\frac{\partial}{\partial\eps}\left[1-\left(\eps+{\P{k}}/{K}\right)^2\right]^{-\frac{1}{2}}
 -\displaystyle\frac{\partial}{\partial\eps}\left[1-\left(\eps+{\P{j}}/{K}\right)^2\right]^{-\frac{1}{2}}\\
 &=&\displaystyle\sum_{k\neq j}\left[1-\left(\eps+{\P{k}}/{K}\right)^2\right]^{-\frac{3}{2}}\left[1-\left(\eps+{\P{j}}/{K}\right)^2\right]^{-1}
 \end{array}
 \\
 \begin{array}{rcr}
 \phantom{\frac{\partial}{\partial\eps}\sum_kc_k^{-1}}&\phantom{=}&\times\left[{1+\left(\eps+{\P{j}}/{K}\right)\left(\eps+{\P{k}}/{K}\right)}\right]\left({\P{k}-\P{j}}\right)/{K}\, .
 \end{array}
\end{align*}
The only term in the last expression that is not necessarily positive is $\left(\P{k}-\P{j}\right)/{K}$. But if $j=j_0$ (resp. $j=j_1$), this term is always positive (resp. negative) for $k\neq j$, 
and consequently the whole sum is positive (resp. negative). Thus, following a connected component of the level curve $\A_j=L$,  whenever $\sum_kc_k^{-1}$ hits zero, its derivative always has the 
same sign, therefore, by continuity, it cannot cross zero more than once. This completes the proof.
\end{proof}

\paragraph*{}
The proof of Theorem~\ref{mainThm} finally relies on Taylor's Lemma 2.1~\cite{Tay12}, which we recall here.

\begin{lem}[Taylor~\cite{Tay12}]\label{lemTaylor}
 Let $\{\theta_i^{(0)}\}$ be any stable solution of the power flow Eq.~(\ref{PFsimple}) on any network. 
 Then for any non-empty node subset $S$,
 \begin{align*}
  \sum_{\substack{\langle ij\rangle\colon \\ i\in S,j\notin S}}\cos(\Delta_{ij}^{(0)})\geq0\, .
 \end{align*}
\end{lem}

\paragraph*{}
In other words, if we can partition the nodes of the network in two sets $S$ and $S^c$, such that the sum of cosines of the angle differences on all the lines between these two sets 
is smaller than $0$, then the solution is unstable. In our case of a cycle network, if the angle differences on two lines are larger than ${\pi}/{2}$ or less than $-{\pi}/{2}$, removing 
these two lines splits the network in two parts, $S$ and $S^c$, such that
\begin{align*}
  \sum_{\substack{\langle ij\rangle\colon \\ i\in S,j\notin S}}\cos(\Delta_{ij}^{(0)})<0\, ,
\end{align*}
and the solution is unstable. We conclude that there is at most a single $|\Delta_{i,i+1}|>\pi/2$. We are now ready to prove Theorem~\ref{mainThm}.

\begin{remark}
 Instead of Taylor's Lemma 2.1, we could use the necessary condition for stability 
 of Ref.~\onlinecite{Do12}, that 
 if $\{\theta_i^{(0)}\}$ is a stable solution of the power flow Eq.~\eqref{PFsimple}, 
 then there exists a spanning tree $T$ of the network such that for all edges $e\in E_T$,
 \begin{align*}
  \cos(\Delta_e^{(0)})&\geq 0\, .
 \end{align*}
 Taylor's lemma seems to be slightly more general. As a matter of fact, 
 it is an easy exercise to construct an example of a weighted graph containing a positively weighted    
 spanning tree, but
 such that there exists a non-empty node subset $S$ with 
 \begin{align*}
  \sum_{\substack{\langle ij\rangle\colon \\ i\in S,j\notin S}}\cos(\Delta_{ij}^{(0)})<0\, ,
 \end{align*}
 In this case, Taylor's Lemma 2.1 implies instability, while Ref.~\onlinecite{Do12} does not.
\end{remark}

\begin{proof}[Proof of Theorem~\ref{mainThm}]
Since for any $K_0$, $\A_0(K_0,\eps)$ is an increasing function of $\eps$, we know that for any integer $q\in\big[-{n}/{4},{n}/{4}\big]$, the level 
set of $\A_0=2\pi q$ is a single level curve. Furthermore, any point on such a level curve corresponds to a stable solution of the 
power flow Eq.~(\ref{PFsimple}). Starting from large values of $K$ and following this level curve while decreasing $K$, 
Corollary~\ref{shape} implies that it meets the boundary of $\D$ at some point. Assume that it meets the upper boundary 
$\D_1$ at $X=(K^*,\eps^*)$ as shown on Fig.~\ref{pseudo_parabola} 
(the case of the lower boundary $\D_0$ is treated in the same way, interchanging $j_1$ and $j_0$ in what follows). 
We know that $\A_0(K^*,\eps^*)=\A_{j_1}(K^*,\eps^*)$. As $\A_{j_1}$ is monotonous on $\D_1$ and smooth in the interior 
of the domain $\D$, there is a level curve of $\A_{j_1}=2\pi q$ starting at $X$ (the red line in Fig.~\ref{pseudo_parabola}), 
and by Lemma~\ref{finally}, it is unique. Furthermore, at this point, the corresponding solution is stable.
\begin{figure}[b]
 \begin{center}
  \includegraphics[width=300px]{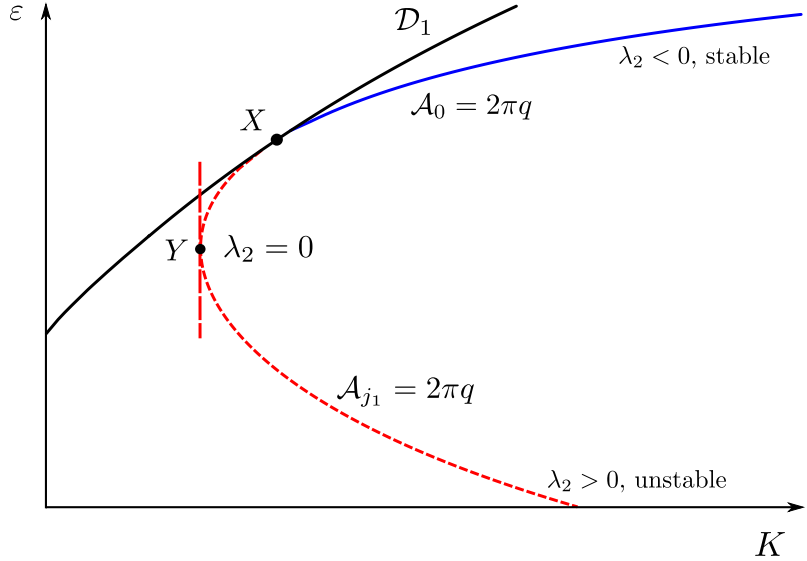}
  \caption{\textit{Level curves $\A_0(\eps,K)=2\pi q$ (blue line) and $\A_{j_1}(\eps,K)=2\pi q$ (dashed red line). Both functions are undefined above the boundary $\D_1$ (black line). 
  The corresponding solutions are stable 
  ($\lambda_2<0$) along the blue curve and between points $X$ and $Y$ on the dashed red curve, and unstable ($\lambda_2>0$) along the dashed red curve, from $Y$ to $K\to\infty$.}}
  \label{pseudo_parabola}
 \end{center}
\end{figure}

\paragraph*{}
According to Corollary~\ref{shape}, the level curve of $\A_{j_1}$ either meets the boundary of $\D$ or goes to $K\to\infty$. First, it cannot meet $\D_1$ because the value of $\A_{j_1}$ is 
strictly increasing with $K$ on $\D_1$ and the level curve cannot be closed by Corollary~\ref{shape}. Second, if it goes to $K\to\infty$, we know from Section~\ref{infK} that for $K$ 
large enough, the solution is unstable. Third, if it meets $\D_0$, Lemma~\ref{lemTaylor} implies that at this point, the corresponding solution is unstable, because at this point, 
$|\Delta_{j_1,j_1+1}|>\pi/2$ and $|\Delta_{j_0,j_0+1}|=\pi/2$. 
Thus along the level curve considered, the eigenvalue $\lambda_2$ has to change sign. 
Following Lemmas~\ref{bla} and \ref{blabla}, this happens only once, at point $Y$ shown on Fig.~\ref{pseudo_parabola} where the level curve changes direction with respect to $K$.

\paragraph*{}
Assume now that there is another connected component of the level set of $\A_{j_1}=2\pi q$. From Corollary~\ref{shape} it cannot be closed and by monotonicity of $\A_{j_1}$ 
along $\D_1$ and Lemma~\ref{finally}, it cannot meet the upper boundary $\D_1$. Thus the corresponding solutions are unstable at both ends of this level curve and as, by Lemmas~\ref{bla} 
and \ref{blabla}, $\lambda_2$ changes sign at most once along a level curve, then the corresponding solutions are unstable all along this level curve.

\paragraph*{}
We conclude that the number of $\eps$ values corresponding to stable solutions of the power flow equations increases with $K$, because a solution appears at point $Y$ 
and exists for any larger $K$.
\end{proof}
  
\begin{remark}
If there are two indices $i_1$ and $i_2$ such that $P_{i_1,i_1+1}^*=P_{i_2,i_2+1}^*=P_{\min}^*$ (the same works with $P_{\max}^*$), 
then $\sum_kc_k^{-1}>0$ for $a_{i_2}\in\big(-\pi,-{\pi}/{2}\big)\cup\big({\pi}/{2},\pi\big]$ and all other $a_i$'s in 
$\big[-{\pi}/{2},{\pi}/{2}\big]$, because $c_{i_1}=-c_{i_2}$ and then 
\begin{align*}
 \sum_kc_k^{-1}&=\sum_{k\neq i_1,i_2}c_k^{-1}>0\, .
\end{align*}
Hence, inside $\D$, $\lambda_2$ never changes sign along the level curves of $\A_{i_1}$ and $\A_{i_2}$.
This result together with the fact that for $K\rightarrow\infty$ the solutions corresponding to the level curves $\A_{i_1}$ and $\A_{i_2}$ 
are known to be unstable implies that such solutions remain unstable also for finite values of $K$.
Which implies that, in this case, no solution having one angle difference outside the interval $[-\pi/2,\pi/2]$
can be locally stable.
\end{remark}

\paragraph*{}
To summarize, we showed that while decreasing $K$, $\mathcal{N}(K)$ also decreases, and that any stable solution of Eq.~\eqref{PFsimple} for finite $K$ 
is a continuation of a solution for $K\to\infty$. 
We also showed that for finite $K$, stable solutions have at most one angle difference outside $[-\pi/2,\pi/2]$, and that such solutions are continuations of 
solutions with all angle differences in $[-\pi/2,\pi/2]$. 
Fig.~\ref{globalPicture} illustrates the whole situation. The domain $\D$ is bounded above by the curve $\D_1$ and below by $\D_0$.
The blue lines are the $2\pi q$-level curves of $\A_0$ for $q\in\{-1,0,1\}$, i.e. any point on a blue curve gives a pair of values $(K,\eps)$ corresponding to a 
stable solution of Eq.~\eqref{PFsimple} with all angle differenes in $[-\pi/2,\pi/2]$. 
The red dashed lines and the green dash-dotted line are the $2\pi q$-level curves of $\A_{j_1}$ and $\A_{j_0}$ respectively.
The points on the red dashed curves correspond to solutions (not necessarily stable) where the angle 
difference on the line carrying $P_{\max}^*$ is in $(-\pi,-\pi/2)\cup(\pi/2,\pi]$ and the points on the green dash-dotted curve correspond to solutions where the angle 
difference along the line carrying $P_{\min}^*$ is in $(-\pi,-\pi/2)\cup(\pi/2,\pi]$. Any blue line meets either a red dashed line on $\D_1$ 
(a zoom-in of this is depicted in Fig.~\ref{pseudo_parabola}) or green dash-dotted line on $\D_0$.
While increasing $K$, stable solutions appear on the level curves of $\A_{j_1}$ and $\A_{j_0}$ (at point $Y$ in Fig.~\ref{pseudo_parabola}), 
thus with one angle difference larger than $\pi/2$ (or less than $-\pi/2$). 
This angle difference then enters $[-\pi/2,\pi/2]$ while $K$ increases. This happens 
at point $X$ in Fig.~\ref{pseudo_parabola}. 
Then the stable solution persists for any larger $K$ along the corresponding level curve of $\A_0$.
\begin{figure}[b]
 \begin{center}
  \includegraphics[width=16cm]{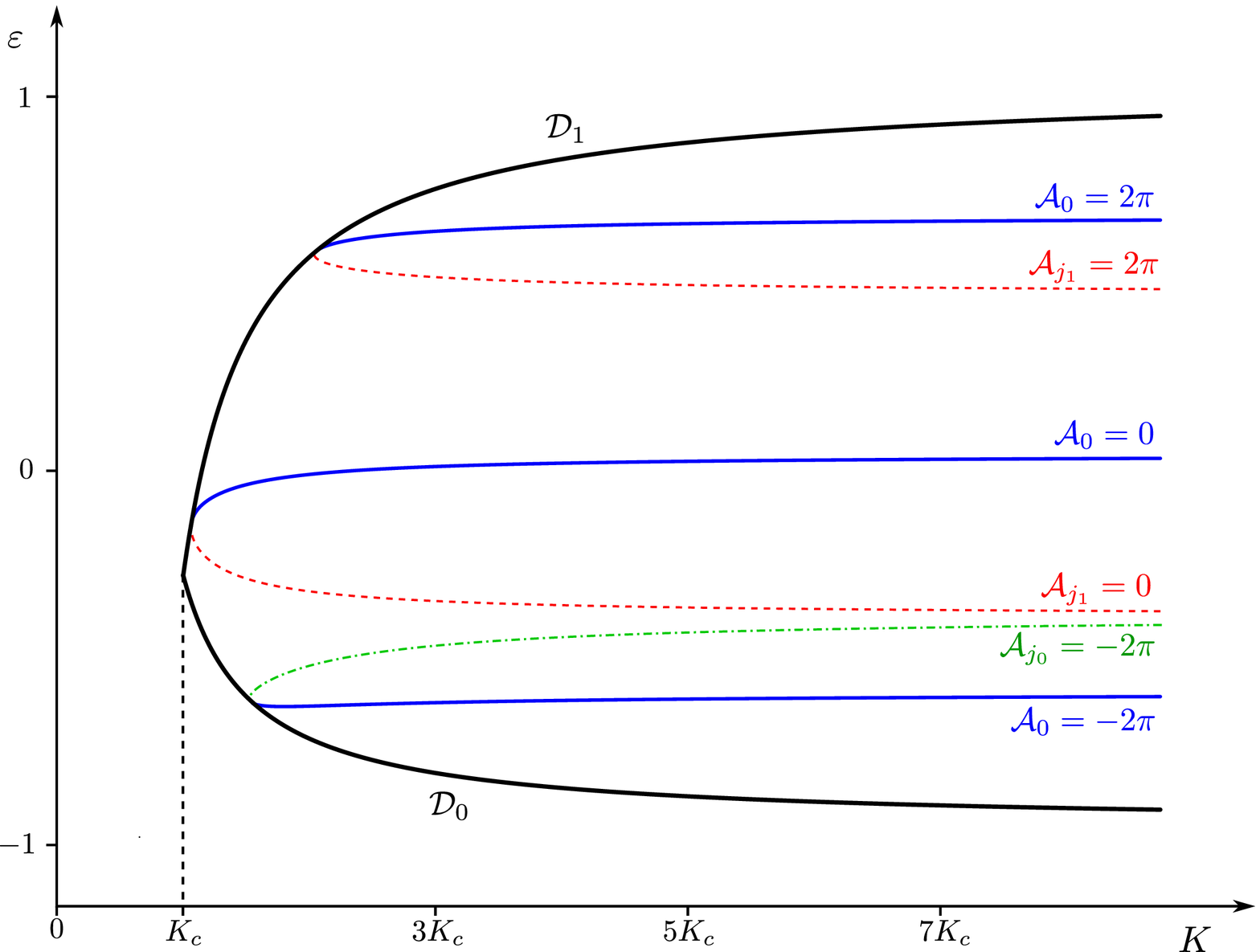}
  \caption{\textit{Level curves $\A_0=2\pi q$ (blue), $\A_{j_1}=2\pi q$ (dashed red) and $\A_{j_0}=2\pi q$ (dash-dotted green), for different $q$-values, in the $(K,\eps)$-plane.
  The level curves of $\A_0$ and $\A_{j_1}$ meet on the upper boundary $\D_1$ of 
  the domain $\D$ defined in Eq.~(\protect\ref{eq:D}), 
  and the level curves of $\A_0$ and $\A_{j_0}$ meet on the lower boundary $\D_0$. 
  The two boundaries $\D_0$ and $\D_1$ meet at $K_c=(P_{\max}^*-P_{\min}^*)/2$. A zoom-in 
  of the region where the level curves $\A_0=2\pi$ and $\A_{j_1}=2\pi$ meet is depicted in 
  Fig.~\protect\ref{pseudo_parabola}}}
  \label{globalPicture}
 \end{center}
\end{figure}

\section{Conclusion}\label{conclusion}
We have investigated the multiplicity of stable stationary solutions to the Kuramoto model. 
For any network, Theorem~\ref{thm_loop_flow} shows that any two different solutions of 
Eqs.~\eqref{eq:Ksteady} and \eqref{PFsimple} differ by a 
combination of circulating flows around the cycles of the network.
We showed that these loop flows are quantized and labelled by a topological winding number. 
In the particular case of single-cycle networks, we then derived an upper bound on the number of stable solutions of the power flow Eq.~(\ref{PFsimple}), 
\begin{align*}
 \mathcal{N}&\leq2\, {\rm Int}\left[{n}/{4}\right]+1\, ,
\end{align*}
which is algebraic in $n$, the length of the cycle. It significantly improves the exponential bounds obtained if Refs.~\onlinecite{Ngu14,Meh14}. 
Our result generalizes the bounds obtained by Ochab and G\'{o}ra~\cite{Och10}, dealing in particular with angle differences 
larger than $\pi/2$, and extends the results of Rogge and Aeyles~\cite{Rog04} to bidirectional couplings.

\paragraph*{}
As parallel results, we obtained some sharp conditions for the solutions on a cycle network 
with some angle differences in $\big(-\pi,-{\pi}/{2}\big)\cup\big({\pi}/{2},\pi\big]$ to be stable. 
We showed that at most one angle difference can be larger than $\pi/2$ in a stable solution and 
it can only be the case on the most loaded line. 
Moreover, any stable solution with an angle difference larger than $\pi/2$ can be directly connected 
to a solution with all angle differences 
in $[-\pi/2,\pi/2]$ for the same network at larger $K$.

\paragraph*{}
The quantized loop flows discussed above are highly undesirable in electric power grids.
They transmit power which is never distributed but only generates ohmic losses.
A deeper understanding of loop flows, how they appear and how to make them 
disappear could greatly help in devising power grids protected against their emergence. In all
likelihood, this would be of great interest for power grid operators. 

\paragraph*{}
Another line of possible future research would be to compare the stability of different solutions. 
This could be done in at least two ways, first,
comparing the spectra of the stability matrices for different solutions, second, comparing
the volumes of the respective basins of attraction. This second approach was proposed in 
Ref.~\onlinecite{Wil06}. Of particular interest would be to relate these two measures of stability with 
winding numbers.

\paragraph*{}
Obviously, the next step in this investigation is to study how our results can be extended to more general networks with multiple cycles. 
It is clear that in the case of independent cycles as in Fig.~\ref{indepCycles}, the number of stable solutions is bounded by
\begin{align*}
 \mathcal{N}=\left(2\, {\rm Int}\left[{n_1}/4\right]+1\right)\left(2\, {\rm Int}\left[{n_2}/4\right]+1\right)\, ,
\end{align*}
where $n_1$ and $n_2$ are the number of edges in the two cycles respectively. 
In this case, a loop flow on one of the cycles does not influence the loop flow on the other cycle.
The problem becomes more intricate when we have to deal with cycles sharing edges, where the loop flows add, Fig.~\ref{depCycles}. 
Here the flow on one of the cycles limits the loop flow on the other cycle because it could saturate the capacity of the common lines. Work along those lines is in progress.
\begin{figure}[b]
 \centering
 \subfloat[]{
  \includegraphics[width=175px]{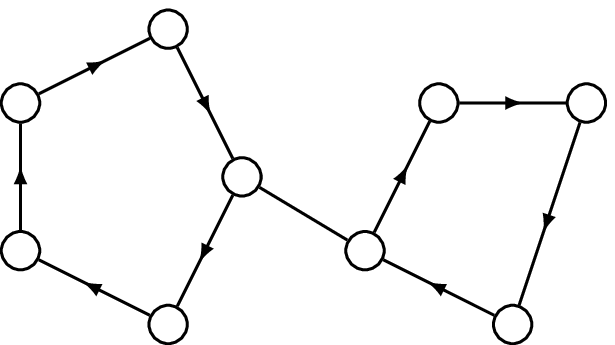}
  \label{indepCycles}
  }
  \hspace{2cm}
  \subfloat[]{
   \includegraphics[width=150px]{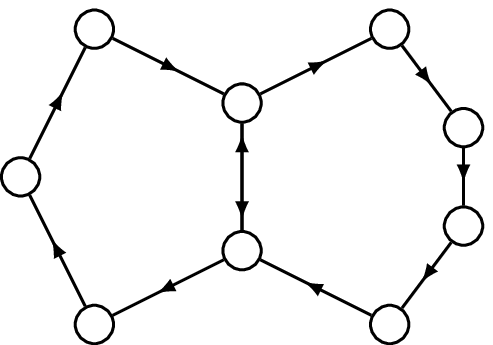}
  \label{depCycles}
 }
 \caption{\textit{Two graphs with two cycles. Left: the cycles are independent but connected, the loop flow on one of them does not influence the loop flow on the other one. 
 Right: the cycles are not independent. Thus having a loop flow on one of them restricts the possible flows on the other one.}}
\end{figure}

\section{Acknowledgments}
This work has been supported by the Swiss National Science Foundation under an AP Energy Grant.
The authors would like thank L. Pagnier for useful discussions.

\begin{appendices}
\numberwithin{equation}{section}

\section{{$LU$}-decomposition of the stability matrix}\label{induction}

We prove inductively that the diagonal elements $\beta_i$ of the upper triangular matrix of the $LU$-factorization of 
the stability matrix, obtained through Thomas algorithm~\cite{Gan14}, satisfy the relation
\begin{align*}
 \mu_{i-1}\cdot\beta_i&=-\mu_i\, ,
\end{align*}
where $\mu_i$'s are defined in Eq.~(\ref{eq:Definition mu}).
For $i=1$, we have
\begin{align*}
 \mu_0\cdot\beta_1&=1\cdot (-c_1-c_2)\\
 &=-\mu_1\, .
\end{align*}
Suppose now that $\mu_{i-1}\cdot\beta_i=-\mu_i$. Let us show that $\mu_i\cdot\beta_{i+1}=-\mu_{i+1}$.
\begin{align*}
 \mu_i\cdot\beta_{i+1}&=\mu_i\left(-c_{i+1}^\iPower-c_{i+2}^\iPower-\frac{c_{i+1}^2}{\beta_i}\right)\\
 &=-c_{i+1}^\iPower\mu_i-c_{i+2}^\iPower\mu_i+c_{i+1}^2\mu_{i-1}\\
 &=-c_{i+1}^\iPower\sumcos{i+1}-c_{i+2}^\iPower\sumcos{i+1}+c_{i+1}^2\sumcos{i}\\
 &=-c_{i+1}^\iPower\left(c_{i+1}^\iPower\sumcos{i}+\prod_{k=1}^ic_k^\iPower\right)-c_{i+2}^\iPower\sumcos{i+1}+c_{i+1}^2\sumcos{i}\\
 &=-c_{i+1}^\iPower\prod_{k=1}^{i}c_k^\iPower-c_{i+2}^\iPower\sumcos{i+1}\\
 &=-\sumcos{i+2}\\
 &=-\mu_{i+1}\, .
\end{align*}

\end{appendices}

\bibliographystyle{apsrev4-1}
\bibliography{biblio_mlfs}

\begin{thebibliography}{41}%
\makeatletter
\providecommand \@ifxundefined [1]{%
 \@ifx{#1\undefined}
}%
\providecommand \@ifnum [1]{%
 \ifnum #1\expandafter \@firstoftwo
 \else \expandafter \@secondoftwo
 \fi
}%
\providecommand \@ifx [1]{%
 \ifx #1\expandafter \@firstoftwo
 \else \expandafter \@secondoftwo
 \fi
}%
\providecommand \natexlab [1]{#1}%
\providecommand \enquote  [1]{``#1''}%
\providecommand \bibnamefont  [1]{#1}%
\providecommand \bibfnamefont [1]{#1}%
\providecommand \citenamefont [1]{#1}%
\providecommand \href@noop [0]{\@secondoftwo}%
\providecommand \href [0]{\begingroup \@sanitize@url \@href}%
\providecommand \@href[1]{\@@startlink{#1}\@@href}%
\providecommand \@@href[1]{\endgroup#1\@@endlink}%
\providecommand \@sanitize@url [0]{\catcode `\\12\catcode `\$12\catcode
  `\&12\catcode `\#12\catcode `\^12\catcode `\_12\catcode `\%12\relax}%
\providecommand \@@startlink[1]{}%
\providecommand \@@endlink[0]{}%
\providecommand \url  [0]{\begingroup\@sanitize@url \@url }%
\providecommand \@url [1]{\endgroup\@href {#1}{\urlprefix }}%
\providecommand \urlprefix  [0]{URL }%
\providecommand \Eprint [0]{\href }%
\providecommand \doibase [0]{http://dx.doi.org/}%
\providecommand \selectlanguage [0]{\@gobble}%
\providecommand \bibinfo  [0]{\@secondoftwo}%
\providecommand \bibfield  [0]{\@secondoftwo}%
\providecommand \translation [1]{[#1]}%
\providecommand \BibitemOpen [0]{}%
\providecommand \bibitemStop [0]{}%
\providecommand \bibitemNoStop [0]{.\EOS\space}%
\providecommand \EOS [0]{\spacefactor3000\relax}%
\providecommand \BibitemShut  [1]{\csname bibitem#1\endcsname}%
\let\auto@bib@innerbib\@empty
\bibitem [{\citenamefont {Strogatz}(2003)}]{Str03}%
  \BibitemOpen
  \bibfield  {author} {\bibinfo {author} {\bibfnamefont {S.~H.}\ \bibnamefont
  {Strogatz}},\ }\href@noop {} {\emph {\bibinfo {title} {SYNC, the Emerging
  Science of Spontaneous Order}}}\ (\bibinfo  {publisher} {Penguin Books},\
  \bibinfo {address} {London},\ \bibinfo {year} {2003})\BibitemShut {NoStop}%
\bibitem [{\citenamefont {Winfree}(1967)}]{Win67}%
  \BibitemOpen
  \bibfield  {author} {\bibinfo {author} {\bibfnamefont {A.~T.}\ \bibnamefont
  {Winfree}},\ }\href@noop {} {\bibfield  {journal} {\bibinfo  {journal} {J.
  Theoret. Biol.}\ }\textbf {\bibinfo {volume} {16}},\ \bibinfo {pages} {15}
  (\bibinfo {year} {1967})}\BibitemShut {NoStop}%
\bibitem [{\citenamefont {Kuramoto}(1984)}]{Kur84}%
  \BibitemOpen
  \bibfield  {author} {\bibinfo {author} {\bibfnamefont {Y.}~\bibnamefont
  {Kuramoto}},\ }\href@noop {} {\bibfield  {journal} {\bibinfo  {journal}
  {Progr. Theoret. Phys. Suppl.}\ }\textbf {\bibinfo {volume} {79}},\ \bibinfo
  {pages} {223} (\bibinfo {year} {1984})}\BibitemShut {NoStop}%
\bibitem [{\citenamefont {Strogatz}(2000)}]{Str00}%
  \BibitemOpen
  \bibfield  {author} {\bibinfo {author} {\bibfnamefont {S.~H.}\ \bibnamefont
  {Strogatz}},\ }\href@noop {} {\bibfield  {journal} {\bibinfo  {journal}
  {Physica D}\ }\textbf {\bibinfo {volume} {143}},\ \bibinfo {pages} {1}
  (\bibinfo {year} {2000})}\BibitemShut {NoStop}%
\bibitem [{\citenamefont {Acebr\'on}\ \emph {et~al.}(2005)\citenamefont
  {Acebr\'on}, \citenamefont {Bonilla}, \citenamefont {P\'erez~Vicente},
  \citenamefont {Ritort},\ and\ \citenamefont {Spigler}}]{Ace05}%
  \BibitemOpen
  \bibfield  {author} {\bibinfo {author} {\bibfnamefont {J.~A.}\ \bibnamefont
  {Acebr\'on}}, \bibinfo {author} {\bibfnamefont {L.~L.}\ \bibnamefont
  {Bonilla}}, \bibinfo {author} {\bibfnamefont {C.~J.}\ \bibnamefont
  {P\'erez~Vicente}}, \bibinfo {author} {\bibfnamefont {F.}~\bibnamefont
  {Ritort}}, \ and\ \bibinfo {author} {\bibfnamefont {R.}~\bibnamefont
  {Spigler}},\ }\href@noop {} {\bibfield  {journal} {\bibinfo  {journal} {Rev.
  Mod. Phys.}\ }\textbf {\bibinfo {volume} {77}},\ \bibinfo {pages} {137}
  (\bibinfo {year} {2005})}\BibitemShut {NoStop}%
\bibitem [{\citenamefont {D\"orfler}\ and\ \citenamefont
  {Bullo}(2014)}]{Dor14}%
  \BibitemOpen
  \bibfield  {author} {\bibinfo {author} {\bibfnamefont {F.}~\bibnamefont
  {D\"orfler}}\ and\ \bibinfo {author} {\bibfnamefont {F.}~\bibnamefont
  {Bullo}},\ }\href@noop {} {\bibfield  {journal} {\bibinfo  {journal}
  {Automatica}\ }\textbf {\bibinfo {volume} {50}},\ \bibinfo {pages} {1539}
  (\bibinfo {year} {2014})}\BibitemShut {NoStop}%
\bibitem [{\citenamefont {Ermentrout}(1985)}]{Erm85}%
  \BibitemOpen
  \bibfield  {author} {\bibinfo {author} {\bibfnamefont {G.~B.}\ \bibnamefont
  {Ermentrout}},\ }\href@noop {} {\bibfield  {journal} {\bibinfo  {journal} {J.
  Math. Biol.}\ }\textbf {\bibinfo {volume} {22}},\ \bibinfo {pages} {1}
  (\bibinfo {year} {1985})}\BibitemShut {NoStop}%
\bibitem [{\citenamefont {van Hemmen}\ and\ \citenamefont
  {Wreskinski}(1993)}]{Hem93}%
  \BibitemOpen
  \bibfield  {author} {\bibinfo {author} {\bibfnamefont {J.~L.}\ \bibnamefont
  {van Hemmen}}\ and\ \bibinfo {author} {\bibfnamefont {W.~F.}\ \bibnamefont
  {Wreskinski}},\ }\href@noop {} {\bibfield  {journal} {\bibinfo  {journal} {J.
  Stat. Phys.}\ }\textbf {\bibinfo {volume} {72}},\ \bibinfo {pages} {145}
  (\bibinfo {year} {1993})}\BibitemShut {NoStop}%
\bibitem [{\citenamefont {Aeyels}\ and\ \citenamefont {Rogge}(2004)}]{Aey04}%
  \BibitemOpen
  \bibfield  {author} {\bibinfo {author} {\bibfnamefont {D.}~\bibnamefont
  {Aeyels}}\ and\ \bibinfo {author} {\bibfnamefont {J.~A.}\ \bibnamefont
  {Rogge}},\ }\href@noop {} {\bibfield  {journal} {\bibinfo  {journal} {Prog.
  Th. Phys.}\ }\textbf {\bibinfo {volume} {112}},\ \bibinfo {pages} {921}
  (\bibinfo {year} {2004})}\BibitemShut {NoStop}%
\bibitem [{\citenamefont {Mirollo}\ and\ \citenamefont
  {Strogatz}(2005)}]{Mir05}%
  \BibitemOpen
  \bibfield  {author} {\bibinfo {author} {\bibfnamefont {R.~E.}\ \bibnamefont
  {Mirollo}}\ and\ \bibinfo {author} {\bibfnamefont {S.~H.}\ \bibnamefont
  {Strogatz}},\ }\href@noop {} {\bibfield  {journal} {\bibinfo  {journal}
  {Physica D}\ }\textbf {\bibinfo {volume} {205}},\ \bibinfo {pages} {249}
  (\bibinfo {year} {2005})}\BibitemShut {NoStop}%
\bibitem [{\citenamefont {Taylor}(2012)}]{Tay12}%
  \BibitemOpen
  \bibfield  {author} {\bibinfo {author} {\bibfnamefont {R.}~\bibnamefont
  {Taylor}},\ }\href {http://iopscience.iop.org/1751-8121/45/5/055102}
  {\bibfield  {journal} {\bibinfo  {journal} {J. Phys. A}\ }\textbf {\bibinfo
  {volume} {45}},\ \bibinfo {pages} {055102} (\bibinfo {year}
  {2012})}\BibitemShut {NoStop}%
\bibitem [{\citenamefont {Korsak}(1972)}]{Kor72}%
  \BibitemOpen
  \bibfield  {author} {\bibinfo {author} {\bibfnamefont {A.~J.}\ \bibnamefont
  {Korsak}},\ }\href
  {http://ieeexplore.ieee.org/xpls/abs_all.jsp?arnumber=4074824} {\bibfield
  {journal} {\bibinfo  {journal} {IEEE Trans. Power App. Syst.}\ }\textbf
  {\bibinfo {volume} {PAS-91}},\ \bibinfo {pages} {1093} (\bibinfo {year}
  {1972})}\BibitemShut {NoStop}%
\bibitem [{\citenamefont {Bergen}\ and\ \citenamefont {Vittal}(2000)}]{Ber00}%
  \BibitemOpen
  \bibfield  {author} {\bibinfo {author} {\bibfnamefont {A.~R.}\ \bibnamefont
  {Bergen}}\ and\ \bibinfo {author} {\bibfnamefont {V.}~\bibnamefont
  {Vittal}},\ }\href@noop {} {\emph {\bibinfo {title} {Power {Systems}
  {Analysis}}}}\ (\bibinfo  {publisher} {Prentice Hall},\ \bibinfo {year}
  {2000})\BibitemShut {NoStop}%
\bibitem [{\citenamefont {Tamura}\ \emph {et~al.}(1983)\citenamefont {Tamura},
  \citenamefont {Mori},\ and\ \citenamefont {Iwamoto}}]{Tam83}%
  \BibitemOpen
  \bibfield  {author} {\bibinfo {author} {\bibfnamefont {Y.}~\bibnamefont
  {Tamura}}, \bibinfo {author} {\bibfnamefont {H.}~\bibnamefont {Mori}}, \ and\
  \bibinfo {author} {\bibfnamefont {S.}~\bibnamefont {Iwamoto}},\ }\href
  {\doibase 10.1109/TPAS.1983.318052} {\bibfield  {journal} {\bibinfo
  {journal} {IEEE Trans. Power App. Syst.}\ }\textbf {\bibinfo {volume}
  {PAS-102}},\ \bibinfo {pages} {1115} (\bibinfo {year} {1983})}\BibitemShut
  {NoStop}%
\bibitem [{\citenamefont {Klos}\ and\ \citenamefont {Wojcicka}(1991)}]{Klos91}%
  \BibitemOpen
  \bibfield  {author} {\bibinfo {author} {\bibfnamefont {A.}~\bibnamefont
  {Klos}}\ and\ \bibinfo {author} {\bibfnamefont {J.}~\bibnamefont
  {Wojcicka}},\ }\href {\doibase
  http://dx.doi.org/10.1016/0142-0615(91)90050-6} {\bibfield  {journal}
  {\bibinfo  {journal} {Int. J. Elect. Power Energy Syst.}\ }\textbf {\bibinfo
  {volume} {13}},\ \bibinfo {pages} {268 } (\bibinfo {year}
  {1991})}\BibitemShut {NoStop}%
\bibitem [{\citenamefont {Baillieul}\ and\ \citenamefont
  {Byrnes}(1982)}]{Bai82}%
  \BibitemOpen
  \bibfield  {author} {\bibinfo {author} {\bibfnamefont {J.}~\bibnamefont
  {Baillieul}}\ and\ \bibinfo {author} {\bibfnamefont {C.~I.}\ \bibnamefont
  {Byrnes}},\ }\href@noop {} {\bibfield  {journal} {\bibinfo  {journal} {IEEE
  Trans. Circuits Syst.}\ }\textbf {\bibinfo {volume} {29}},\ \bibinfo {pages}
  {724} (\bibinfo {year} {1982})}\BibitemShut {NoStop}%
\bibitem [{\citenamefont {Nguyen}\ and\ \citenamefont
  {Turitsyn}(2014)}]{Ngu14}%
  \BibitemOpen
  \bibfield  {author} {\bibinfo {author} {\bibfnamefont {H.~D.}\ \bibnamefont
  {Nguyen}}\ and\ \bibinfo {author} {\bibfnamefont {K.~S.}\ \bibnamefont
  {Turitsyn}},\ }in\ \href {\doibase 10.1109/PESGM.2014.6938797} {\emph
  {\bibinfo {booktitle} {PES General Meeting - Conference Exposition, 2014
  IEEE}}}\ (\bibinfo {year} {2014})\BibitemShut {NoStop}%
\bibitem [{\citenamefont {Mehta}\ \emph {et~al.}(2015)\citenamefont {Mehta},
  \citenamefont {Daleo}, \citenamefont {D\"orfler},\ and\ \citenamefont
  {Hauenstein}}]{Meh14}%
  \BibitemOpen
  \bibfield  {author} {\bibinfo {author} {\bibfnamefont {D.}~\bibnamefont
  {Mehta}}, \bibinfo {author} {\bibfnamefont {N.}~\bibnamefont {Daleo}},
  \bibinfo {author} {\bibfnamefont {F.}~\bibnamefont {D\"orfler}}, \ and\
  \bibinfo {author} {\bibfnamefont {J.~D.}\ \bibnamefont {Hauenstein}},\ }\href
  {http://scitation.aip.org/content/aip/journal/chaos/25/5/10.1063/1.4919696}
  {\bibfield  {journal} {\bibinfo  {journal} {Chaos}\ }\textbf {\bibinfo
  {volume} {25}} (\bibinfo {year} {2015})}\BibitemShut {NoStop}%
\bibitem [{\citenamefont {Janssens}\ and\ \citenamefont
  {Kamagate}(2003)}]{Jan03}%
  \BibitemOpen
  \bibfield  {author} {\bibinfo {author} {\bibfnamefont {N.}~\bibnamefont
  {Janssens}}\ and\ \bibinfo {author} {\bibfnamefont {A.}~\bibnamefont
  {Kamagate}},\ }\href {\doibase
  http://dx.doi.org/10.1016/S0142-0615(03)00017-6} {\bibfield  {journal}
  {\bibinfo  {journal} {Int. J. Elect. Power Energy Syst.}\ }\textbf {\bibinfo
  {volume} {25}},\ \bibinfo {pages} {591 } (\bibinfo {year}
  {2003})}\BibitemShut {NoStop}%
\bibitem [{\citenamefont {Onsager}(1949)}]{Ons49}%
  \BibitemOpen
  \bibfield  {author} {\bibinfo {author} {\bibfnamefont {L.}~\bibnamefont
  {Onsager}},\ }\href@noop {} {\bibfield  {journal} {\bibinfo  {journal} {Nuovo
  Cimento}\ }\textbf {\bibinfo {volume} {6}},\ \bibinfo {pages} {249} (\bibinfo
  {year} {1949})}\BibitemShut {NoStop}%
\bibitem [{\citenamefont {Feynman}(1955)}]{Fey55}%
  \BibitemOpen
  \bibfield  {author} {\bibinfo {author} {\bibfnamefont {R.~P.}\ \bibnamefont
  {Feynman}},\ }\href@noop {} {\bibfield  {journal} {\bibinfo  {journal}
  {Progress in Low Temperature Physics}\ }\textbf {\bibinfo {volume} {1}},\
  \bibinfo {pages} {34} (\bibinfo {year} {1955})}\BibitemShut {NoStop}%
\bibitem [{\citenamefont {Abrikosov}(1957)}]{Abr57}%
  \BibitemOpen
  \bibfield  {author} {\bibinfo {author} {\bibfnamefont {A.~A.}\ \bibnamefont
  {Abrikosov}},\ }\href@noop {} {\bibfield  {journal} {\bibinfo  {journal}
  {Sov. Phys. JETP}\ }\textbf {\bibinfo {volume} {5}},\ \bibinfo {pages} {1174}
  (\bibinfo {year} {1957})}\BibitemShut {NoStop}%
\bibitem [{\citenamefont {Byers}\ and\ \citenamefont {Yang}(1961)}]{Bye61}%
  \BibitemOpen
  \bibfield  {author} {\bibinfo {author} {\bibfnamefont {N.}~\bibnamefont
  {Byers}}\ and\ \bibinfo {author} {\bibfnamefont {C.~N.}\ \bibnamefont
  {Yang}},\ }\href@noop {} {\bibfield  {journal} {\bibinfo  {journal} {Phys.
  Rev. Lett.}\ }\textbf {\bibinfo {volume} {7}},\ \bibinfo {pages} {46}
  (\bibinfo {year} {1961})}\BibitemShut {NoStop}%
\bibitem [{\citenamefont {Matveev}\ \emph {et~al.}(2002)\citenamefont
  {Matveev}, \citenamefont {Larkin},\ and\ \citenamefont {Glazman}}]{Mat02}%
  \BibitemOpen
  \bibfield  {author} {\bibinfo {author} {\bibfnamefont {K.~A.}\ \bibnamefont
  {Matveev}}, \bibinfo {author} {\bibfnamefont {A.~I.}\ \bibnamefont {Larkin}},
  \ and\ \bibinfo {author} {\bibfnamefont {L.~I.}\ \bibnamefont {Glazman}},\
  }\href@noop {} {\bibfield  {journal} {\bibinfo  {journal} {Phys. Rev. Lett.}\
  }\textbf {\bibinfo {volume} {89}},\ \bibinfo {pages} {096802} (\bibinfo
  {year} {2002})}\BibitemShut {NoStop}%
\bibitem [{\citenamefont {Rastelli}\ \emph {et~al.}(2013)\citenamefont
  {Rastelli}, \citenamefont {Pop},\ and\ \citenamefont {Hekking}}]{Ras13}%
  \BibitemOpen
  \bibfield  {author} {\bibinfo {author} {\bibfnamefont {G.}~\bibnamefont
  {Rastelli}}, \bibinfo {author} {\bibfnamefont {I.~M.}\ \bibnamefont {Pop}}, \
  and\ \bibinfo {author} {\bibfnamefont {F.~W.~J.}\ \bibnamefont {Hekking}},\
  }\href@noop {} {\bibfield  {journal} {\bibinfo  {journal} {Phys. Rev. B}\
  }\textbf {\bibinfo {volume} {87}},\ \bibinfo {pages} {174513} (\bibinfo
  {year} {2013})}\BibitemShut {NoStop}%
\bibitem [{\citenamefont {Bukhsh}\ \emph {et~al.}(2013)\citenamefont {Bukhsh},
  \citenamefont {Grothey}, \citenamefont {McKinnon},\ and\ \citenamefont
  {Trodden}}]{Buk13}%
  \BibitemOpen
  \bibfield  {author} {\bibinfo {author} {\bibfnamefont {W.~A.}\ \bibnamefont
  {Bukhsh}}, \bibinfo {author} {\bibfnamefont {A.}~\bibnamefont {Grothey}},
  \bibinfo {author} {\bibfnamefont {K.~I.~M.}\ \bibnamefont {McKinnon}}, \ and\
  \bibinfo {author} {\bibfnamefont {P.~A.}\ \bibnamefont {Trodden}},\ }\href
  {\doibase 10.1109/TPWRS.2013.2274577} {\bibfield  {journal} {\bibinfo
  {journal} {IEEE Trans. Power Syst.}\ }\textbf {\bibinfo {volume} {28}},\
  \bibinfo {pages} {4780} (\bibinfo {year} {2013})}\BibitemShut {NoStop}%
\bibitem [{\citenamefont {Rogge}\ and\ \citenamefont {Aeyels}(2004)}]{Rog04}%
  \BibitemOpen
  \bibfield  {author} {\bibinfo {author} {\bibfnamefont {J.~A.}\ \bibnamefont
  {Rogge}}\ and\ \bibinfo {author} {\bibfnamefont {D.}~\bibnamefont {Aeyels}},\
  }\href {\doibase 10.1088/0305-4470/37/46/004} {\bibfield  {journal} {\bibinfo
   {journal} {J. Phys. A}\ }\textbf {\bibinfo {volume} {37}},\ \bibinfo {pages}
  {11135} (\bibinfo {year} {2004})}\BibitemShut {NoStop}%
\bibitem [{\citenamefont {Ochab}\ and\ \citenamefont {G\'ora}(2010)}]{Och10}%
  \BibitemOpen
  \bibfield  {author} {\bibinfo {author} {\bibfnamefont {J.}~\bibnamefont
  {Ochab}}\ and\ \bibinfo {author} {\bibfnamefont {P.~F.}\ \bibnamefont
  {G\'ora}},\ }\href@noop {} {\bibfield  {journal} {\bibinfo  {journal} {Acta
  Phys. Pol. B Proc. Suppl.}\ }\textbf {\bibinfo {volume} {3}},\ \bibinfo
  {pages} {453} (\bibinfo {year} {2010})}\BibitemShut {NoStop}%
\bibitem [{\citenamefont {Tilles}\ \emph {et~al.}(2011)\citenamefont {Tilles},
  \citenamefont {Ferreira},\ and\ \citenamefont {Cerdeira}}]{Til11}%
  \BibitemOpen
  \bibfield  {author} {\bibinfo {author} {\bibfnamefont {P.~F.~C.}\
  \bibnamefont {Tilles}}, \bibinfo {author} {\bibfnamefont {F.~F.}\
  \bibnamefont {Ferreira}}, \ and\ \bibinfo {author} {\bibfnamefont {H.~A.}\
  \bibnamefont {Cerdeira}},\ }\href
  {http://link.aps.org/doi/10.1103/PhysRevE.83.066206} {\bibfield  {journal}
  {\bibinfo  {journal} {Phys. Rev. E}\ }\textbf {\bibinfo {volume} {83}}
  (\bibinfo {year} {2011})}\BibitemShut {NoStop}%
\bibitem [{\citenamefont {Roy}\ and\ \citenamefont {Lahiri}(2012)}]{Roy12}%
  \BibitemOpen
  \bibfield  {author} {\bibinfo {author} {\bibfnamefont {T.~K.}\ \bibnamefont
  {Roy}}\ and\ \bibinfo {author} {\bibfnamefont {A.}~\bibnamefont {Lahiri}},\
  }\href {\doibase 10.1016/j.chaos.2012.03.004} {\bibfield  {journal} {\bibinfo
   {journal} {Chaos, Solitons \& Fractals}\ }\textbf {\bibinfo {volume} {45}},\
  \bibinfo {pages} {888} (\bibinfo {year} {2012})}\BibitemShut {NoStop}%
\bibitem [{\citenamefont {Wiley}\ \emph {et~al.}(2006)\citenamefont {Wiley},
  \citenamefont {Strogatz},\ and\ \citenamefont {Girvan}}]{Wil06}%
  \BibitemOpen
  \bibfield  {author} {\bibinfo {author} {\bibfnamefont {D.~A.}\ \bibnamefont
  {Wiley}}, \bibinfo {author} {\bibfnamefont {S.~H.}\ \bibnamefont {Strogatz}},
  \ and\ \bibinfo {author} {\bibfnamefont {M.}~\bibnamefont {Girvan}},\ }\href
  {\doibase 10.1063/1.2165594} {\bibfield  {journal} {\bibinfo  {journal}
  {Chaos}\ }\textbf {\bibinfo {volume} {16}},\ \bibinfo {pages} {015103}
  (\bibinfo {year} {2006})}\BibitemShut {NoStop}%
\bibitem [{\citenamefont {Anderson}(1966)}]{And66}%
  \BibitemOpen
  \bibfield  {author} {\bibinfo {author} {\bibfnamefont {P.~W.}\ \bibnamefont
  {Anderson}},\ }\href {\doibase 10.1103/RevModPhys.38.298} {\bibfield
  {journal} {\bibinfo  {journal} {Rev. Mod. Phys.}\ }\textbf {\bibinfo {volume}
  {38}},\ \bibinfo {pages} {298} (\bibinfo {year} {1966})}\BibitemShut
  {NoStop}%
\bibitem [{\citenamefont {Blatter}\ \emph {et~al.}(1994)\citenamefont
  {Blatter}, \citenamefont {Feigel'man}, \citenamefont {Geshkenbein},
  \citenamefont {Larkin},\ and\ \citenamefont {Vinokur}}]{Bla94}%
  \BibitemOpen
  \bibfield  {author} {\bibinfo {author} {\bibfnamefont {G.}~\bibnamefont
  {Blatter}}, \bibinfo {author} {\bibfnamefont {M.~V.}\ \bibnamefont
  {Feigel'man}}, \bibinfo {author} {\bibfnamefont {V.~B.}\ \bibnamefont
  {Geshkenbein}}, \bibinfo {author} {\bibfnamefont {A.~I.}\ \bibnamefont
  {Larkin}}, \ and\ \bibinfo {author} {\bibfnamefont {V.~M.}\ \bibnamefont
  {Vinokur}},\ }\href {\doibase 10.1103/RevModPhys.66.1125} {\bibfield
  {journal} {\bibinfo  {journal} {Rev. Mod. Phys.}\ }\textbf {\bibinfo {volume}
  {66}},\ \bibinfo {pages} {1125} (\bibinfo {year} {1994})}\BibitemShut
  {NoStop}%
\bibitem [{\citenamefont {Fazio}\ and\ \citenamefont {van~der
  Zant}(2001)}]{Faz01}%
  \BibitemOpen
  \bibfield  {author} {\bibinfo {author} {\bibfnamefont {R.}~\bibnamefont
  {Fazio}}\ and\ \bibinfo {author} {\bibfnamefont {H.}~\bibnamefont {van~der
  Zant}},\ }\href {\doibase http://dx.doi.org/10.1016/S0370-1573(01)00022-9}
  {\bibfield  {journal} {\bibinfo  {journal} {Physics Reports}\ }\textbf
  {\bibinfo {volume} {355}},\ \bibinfo {pages} {235 } (\bibinfo {year}
  {2001})}\BibitemShut {NoStop}%
\bibitem [{\citenamefont {Tavora}\ and\ \citenamefont {Smith}(1972)}]{Tav72}%
  \BibitemOpen
  \bibfield  {author} {\bibinfo {author} {\bibfnamefont {C.~J.}\ \bibnamefont
  {Tavora}}\ and\ \bibinfo {author} {\bibfnamefont {O.~J.~M.}\ \bibnamefont
  {Smith}},\ }\href
  {http://ieeexplore.ieee.org/xpls/abs_all.jsp?arnumber=4074831} {\bibfield
  {journal} {\bibinfo  {journal} {IEEE Trans. Power App. Syst.}\ }\textbf
  {\bibinfo {volume} {PAS-91}},\ \bibinfo {pages} {1138} (\bibinfo {year}
  {1972})}\BibitemShut {NoStop}%
\bibitem [{\citenamefont {Horn}\ and\ \citenamefont {Johnson}(1986)}]{Hor85}%
  \BibitemOpen
  \bibfield  {author} {\bibinfo {author} {\bibfnamefont {R.~A.}\ \bibnamefont
  {Horn}}\ and\ \bibinfo {author} {\bibfnamefont {C.~R.}\ \bibnamefont
  {Johnson}},\ }\href@noop {} {\emph {\bibinfo {title} {Matrix Analysis}}}\
  (\bibinfo  {publisher} {Cambridge University Press},\ \bibinfo {address} {New
  York},\ \bibinfo {year} {1986})\BibitemShut {NoStop}%
\bibitem [{\citenamefont {D\"orfler}\ \emph {et~al.}(2013)\citenamefont
  {D\"orfler}, \citenamefont {Chertkov},\ and\ \citenamefont {Bullo}}]{Dor13}%
  \BibitemOpen
  \bibfield  {author} {\bibinfo {author} {\bibfnamefont {F.}~\bibnamefont
  {D\"orfler}}, \bibinfo {author} {\bibfnamefont {M.}~\bibnamefont {Chertkov}},
  \ and\ \bibinfo {author} {\bibfnamefont {F.}~\bibnamefont {Bullo}},\ }\href
  {\doibase 10.1073/pnas.1212134110} {\bibfield  {journal} {\bibinfo  {journal}
  {Proceedings of the National Academy of Sciences}\ }\textbf {\bibinfo
  {volume} {110}},\ \bibinfo {pages} {2005} (\bibinfo {year}
  {2013})}\BibitemShut {NoStop}%
\bibitem [{\citenamefont {Biggs}(1993)}]{Big93}%
  \BibitemOpen
  \bibfield  {author} {\bibinfo {author} {\bibfnamefont {N.}~\bibnamefont
  {Biggs}},\ }\href@noop {} {\emph {\bibinfo {title} {Algebraic graph
  theory}}},\ \bibinfo {edition} {2nd}\ ed.\ (\bibinfo  {publisher} {Cambridge
  University Press},\ \bibinfo {year} {1993})\BibitemShut {NoStop}%
\bibitem [{\citenamefont {Cheney}(2001)}]{Che01}%
  \BibitemOpen
  \bibfield  {author} {\bibinfo {author} {\bibfnamefont {W.}~\bibnamefont
  {Cheney}},\ }\href {\doibase 10.1007/978-1-4757-3559-8} {\emph {\bibinfo
  {title} {Analysis for {Applied} {Mathematics}}}}\ (\bibinfo  {publisher}
  {Springer},\ \bibinfo {address} {New York},\ \bibinfo {year}
  {2001})\BibitemShut {NoStop}%
\bibitem [{\citenamefont {Gander}\ \emph {et~al.}(2014)\citenamefont {Gander},
  \citenamefont {Gander},\ and\ \citenamefont {Kwok}}]{Gan14}%
  \BibitemOpen
  \bibfield  {author} {\bibinfo {author} {\bibfnamefont {W.}~\bibnamefont
  {Gander}}, \bibinfo {author} {\bibfnamefont {M.~J.}\ \bibnamefont {Gander}},
  \ and\ \bibinfo {author} {\bibfnamefont {F.}~\bibnamefont {Kwok}},\ }\href
  {\doibase 10.1007/978-3-319-04325-8} {\emph {\bibinfo {title} {Scientific
  {Computing} - {An} {Introduction} using {Maple} and {MATLAB}}}}\ (\bibinfo
  {publisher} {Springer International Publishing},\ \bibinfo {year}
  {2014})\BibitemShut {NoStop}%
\bibitem [{\citenamefont {Do}\ \emph {et~al.}(2012)\citenamefont {Do},
  \citenamefont {Boccaletti},\ and\ \citenamefont {Gross}}]{Do12}%
  \BibitemOpen
  \bibfield  {author} {\bibinfo {author} {\bibfnamefont {A.-L.}\ \bibnamefont
  {Do}}, \bibinfo {author} {\bibfnamefont {S.}~\bibnamefont {Boccaletti}}, \
  and\ \bibinfo {author} {\bibfnamefont {T.}~\bibnamefont {Gross}},\
  }\href@noop {} {\bibfield  {journal} {\bibinfo  {journal} {Phys. Rev. Lett.}\
  }\textbf {\bibinfo {volume} {108}} (\bibinfo {year} {2012})}\BibitemShut
  {NoStop}%
\end{thebibliography}%

\end{document}